\documentclass[journal]{IEEEtran}
\usepackage{fancyhdr} 
\usepackage{amsthm}
\usepackage{amsmath}
\usepackage{amssymb}
\usepackage{cases}
\usepackage{graphicx}
\usepackage{float}
\usepackage{multicol}
\usepackage{multirow}
\usepackage{textcomp}
\def\BibTeX{{\rm B\kern-.05em{\sc i\kern-.025em b}\kern-.08em
    T\kern-.1667em\lower.7ex\hbox{E}\kern-.125emX}}

\usepackage{cite}
\usepackage{graphicx}
\usepackage{subfigure}
\usepackage{booktabs}
\usepackage{dsfont}
\usepackage{amsfonts}
\usepackage{bm}
\usepackage{cases}
\usepackage{afterpage}
\usepackage{array}
\usepackage{boxedminipage}
\usepackage{CJK}
\usepackage{setspace}
\usepackage{threeparttable}
\usepackage{makecell}
\usepackage{diagbox}

\usepackage{chngcntr}
\theoremstyle{plain}

\usepackage{algorithm}
\usepackage{algcompatible}

\usepackage[english]{babel}
\usepackage{times}
\usepackage{relsize}
\usepackage{url}
\usepackage{mparhack}
\usepackage{authblk}
\usepackage{epstopdf}
\usepackage{array}
\usepackage{hyperref}
\usepackage[font=small]{caption}

\newcounter{subeqn} %
\makeatletter
\@addtoreset{subeqn}{equation}
\makeatother

\usepackage{xcolor}
\usepackage{tikz}

\newtheorem{theorem}{Theorem}
\newtheorem{lemma}[theorem]{Lemma}

\def\CD{{\mathcal D}}
\def\CM{{\mathcal M}}
\def\CN{{\mathcal N}}

\def\CS{{\mathcal S}}

\def\CJ{{\mathcal J}}

\def\CH{{\mathcal H}}

\def\CA{{\mathcal A}}
\def\CK{{\mathcal K}}

\def\CT{{\mathcal T}}

\def\CG{{\mathcal G}}

\def\CO{{\mathcal O}}

\usepackage[normalem]{ulem}

\DeclareMathOperator*{\argmax}{argmax}
\DeclareMathOperator*{\argmin}{argmin}

\graphicspath{{fig/}}

\begin{document}
\captionsetup[figure]{name={Fig.},labelsep=period}
\captionsetup{justification=centering}



\title{Adapting to Dynamic LEO-B5G Systems: Meta-Critic Learning Based  Efficient Resource Scheduling}

\author[ ]{Yaxiong Yuan, \textit{Student Member, IEEE},~~Lei Lei, \textit{Member, IEEE}, Thang X. Vu,  \textit{Member, IEEE}, 
\newline 
Zheng Chang, \textit{Senior Member, IEEE}, Symeon Chatzinotas, \textit{Senior Member, IEEE}, and Sumei Sun, \textit{Fellow, IEEE}
\thanks{The work has been supported by the ERC project AGNOSTIC (742648), by the FNR CORE projects ROSETTA (C17/IS/11632107), ProCAST (C17/IS/11691338) and ASWELL (C19/IS/13718904), and by the FNR bilateral project LARGOS (12173206). (Corresponding author: Lei Lei.)}
\thanks{Yaxiong Yuan, Lei Lei, Thang X. Vu, and Symeon Chatzinotas are with the Interdisciplinary Centre for Security, Reliability and Trust, Luxembourg University, 1855 Kirchberg, Luxembourg (e-mail: yaxiong.yuan@uni.lu; lei.lei@uni.lu; thang.vu@uni.lu; symeon.chatzinotas@uni.lu).}
\thanks{Zheng Chang is with the School of Computer Science and Engineering, University of Electronic Science and Technology of China, Chengdu 610054, China, and also with the Faculty of Information Technology, University of Jyv\"{a}skyl\"{a}, FI-40014 Jyv\"{a}skyl\"{a}, Finland (e-mail: zheng.chang@jyu.fi).}   
\thanks{S. Sun is with the Institute for Infocomm Research, Agency for Science, Technology, and Research, Singapore 138632 (e-mail: sunsm@i2r.a-star.edu.sg).}
\thanks{This work has been submitted to the IEEE for possible publication.  Copyright may be transferred without notice, after which this version may no longer be accessible.}
}

\IEEEtitleabstractindextext{%
\begin{abstract}
Low earth orbit (LEO) satellite-assisted communications have been considered as one of key elements in beyond 5G systems to provide wide coverage and cost-efficient data services. 
Such dynamic space-terrestrial topologies impose exponential increase in the degrees of freedom in network management. 
In this paper, we address two practical issues for an over-loaded LEO-terrestrial system. 
The first challenge is how to efficiently schedule resources to serve the massive number of connected users, such that more data and users can be delivered/served. 
The second challenge is how to make the algorithmic solution more resilient in adapting to dynamic wireless environments.
To address them, we first propose an iterative suboptimal algorithm to provide an offline benchmark.
To adapt to unforeseen variations, we propose an enhanced meta-critic learning algorithm (EMCL), where a hybrid neural network for parameterization and the Wolpertinger policy for action mapping are designed in EMCL. 
The results demonstrate EMCL's effectiveness and fast-response capabilities in over-loaded systems and in adapting to dynamic environments compare to previous actor-critic and meta-learning methods.

\end{abstract}
}
%
%
\maketitle
 \thispagestyle{fancy} 
      \lhead{} 
      \chead{} 
      \rhead{} 
      \lfoot{} 
      \cfoot{} 
      \rfoot{\thepage} 
      \renewcommand{\headrulewidth}{0pt} 
      \renewcommand{\footrulewidth}{0pt} 

 \pagestyle{fancy}
      \rfoot{\thepage}

%
%
\IEEEdisplaynontitleabstractindextext
%
%
%
%
%
%
%
%
%
%
\section{Introduction}
In beyond 5G networks (B5G), the massive number of connected users and their increasing demands for high-data-rate services can lead to overloading of terrestrial base stations (BSs), which in turn results in degraded user experience, e.g., longer delay in requesting data services or lower data rate \cite{LiLEO5G}.
In order to improve the network performance and user experience, the integration of satellites, e.g., low earth orbit (LEO) satellites, and terrestrial systems is considered as a promising solution to provide cost-efficient data services \cite{Oltjon}.
The solutions for terrestrial network optimization and resource management might not be suitable for direct application to integrated satellite-terrestrial systems \cite{YouLLEO2020}.
In the literature, tailored schemes have been investigated to improve the networks' performance.
In \cite{DiBLEO2019}, the authors proposed a user scheduling scheme to maximize the sum-rate and the number of accessed users by utilizing the LEO-based backhaul.
In \cite{LiYLEO2020}, a joint power allocation and user scheduling scheme was proposed to maximize the network throughput in hierarchical LEO systems with the constraint of transmission delay.
In \cite{WangSLEO2019}, the authors developed a joint resource block allocation and power allocation algorithm to maximize the total transmission rate for LEO systems.
It is worth noting that the resource optimization problems in LEO-terrestrial networks are typically combinatorial and non-convex. 
The conventional iterative optimization methods, e.g., in \cite{DiBLEO2019, LiYLEO2020, WangSLEO2019}, are unaffordable for real-time operations due to their high computational complexity.

\subsection{Related Works: State-of-the-art and Limitations}
Towards an efficient solution, various learning techniques have been studied. 
Compared to supervised learning, reinforcement learning (RL) learns the optimal policy from observed samples without preparing labeled data.  
As one of the promising RL methods, deep reinforcement learning (DRL) adopts deep neural networks (DNNs) for parameterization and rapid decision making.
Recent works have applied RL/DRL for resource management in LEO-terrestrial systems \cite{JHLeeLEODRL2020, HeSLEOQL2020, BDengLEODRL2020}.
In \cite{JHLeeLEODRL2020}, to maximize the achievable rate in LEO-assisted relay networks, a DQN-based algorithm was proposed to make the online decisions for link association.
The authors in \cite{HeSLEOQL2020} adopted multi-agent reinforcement learning to minimize the average number of handovers and improve the efficiency of channel utilization for LEO satellite systems.
In \cite{BDengLEODRL2020}, the authors applied an actor-critic (AC) algorithm to LEO resource allocation, such as beam allocation and power control.
The above RL algorithms in practical LEO systems are limited by the following issue.
That is, the performance of a learning model largely depends on the data originated from the experienced samples or the observed environment, but the wireless environment is highly complex and dynamic.
When network parameters vary dramatically, the performance of the learning models can be degraded.
To remedy this, one has to re-collect a large number of training data and re-train the learning models, which is time-consuming and inefficient to adapt to fast variations \cite{LeiMag}.

To address this issue, a variety of studies focus on how to make the learning models quickly respond to dynamic environments.
Transfer learning applies the knowledge acquired from a source learning task to a target learning task to speed up the re-training process and reduce the volume of the collected new data sets \cite{ShenTL}.
The performance of transfer learning is limited by finding correlated tasks.
Another approach, joint learning, aims at obtaining a single model that can be adapted to dynamic environments by optimizing the loss function over multiple tasks \cite{Osvaldolearning}.
Besides, continual learning can also accelerate the adaptation to the new learning task by adding the experienced data from the previous tasks to the re-training data set, thus avoiding completely forgetting previously learned models \cite{SunCL}.
Joint learning and continual learning might have good learning performance on average but have limited generalization abilities when different tasks are highly diversified \cite{Javed}.
In contrast, meta-learning extracts meta-knowledge and achieves good performance for specific tasks without requiring the related source tasks. 
The authors in \cite{FinnMAML} proposed a model-agnostic meta-learning algorithm (MAML) to obtain the model's initial parameters as meta-knowledge to quickly adapt to new tasks.
In \cite{KateMAMLAC}, an algorithm combining actor-critic with MAML (AC-MAML) was developed to learn a new task from fewer experience data sets.
In \cite{Sungmetacritic}, the authors proposed a promising meta-critic learning framework with better performance than conventional AC and AC-MAML.
In \cite{Wulearntosense}, a meta-learning-based adaptive sensing algorithm was proposed, which determines the next most informative sensing location in wireless sensor networks.
In \cite{Parkml}, meta-learning was applied to find a common initialization vector that enables fast training of an autoencoder for the fading channels.
Most of the meta-learning methods were applied in the areas of pattern recognition \cite{FinnMAML}, robotics \cite{KateMAMLAC, Sungmetacritic}, and physical layer communications \cite{Parkml}.
The considered learning tasks in these works are simple and the action space is small, e.g., \cite{Wulearntosense, Parkml}. 
However, when the learning techniques, e.g., DRL, AC-MAML, or meta-critic learning, are applied to address combinatorial optimization problems in a dynamic LEO-terrestrial network, the action space can be huge and the input-output relationships can become more complex. 
These may degrade the efficiency of the above learning methods.

\subsection{Motivations and Contributions}

Moving beyond state-of-the-art, this paper intends to address the following questions: 
\begin{itemize}
\item Which learning technique can lead to higher performance gain in addressing resource management problems for dynamic LEO-terrestrial networks?
\item How to deal with the huge action space and improve the learning efficiency?
\item How to make the learning solutions more adaptive to dynamic environments? 
\end{itemize}

In this study, we design an enhanced meta-critic learning algorithm (EMCL) for dynamic LEO-terrestrial systems.
To the best of our knowledge, this is arguably the first work to present meta-critic learning to address resource scheduling problems and emphasize the adaptation to non-ideal dynamic environments.
The major contributions are summarized as follows:
\begin{itemize}
\item We design a tailored metric for over-loaded LEO systems with dense user distribution, aiming at serving more users and delivering a higher volume of requested data.
\item We formulate the resource scheduling problem as a quadratic integer programming (QIP) and provides two offline optimization-based benchmarks, i.e., optimal branch and bound (B\&B) algorithm and suboptimal alternating direction method of multipliers-based heuristic algorithm (ADMM-HEU).
\item  Due to the combinatorial nature and the high complexity of the offline solutions, 
we solve the problem from the perspective of DRL by reformulating the problem to a Markov decision process (MDP) to make online decisions intelligently.
\item To enhance the adaptation to dynamic environments, we propose an EMCL algorithm based on the meta-critic framework, where the critic has a good generalization ability to evaluate the actors for new tasks such that the learning agent can adjust the policy timely when the environment changes.
Compared to previous learning methods, the novelty stems from: 1) The tailored design of a  hybrid neural network to extract the features from the current and historical samples; 2)  The integrated Wolpertinger policy to allow the actor to make decisions more efficiently in an exponentially increasing action space. 
\item 
To identify a promising solution, we evaluate the proposed EMCL with other benchmarks in three practical dynamic scenarios, i.e., bursty user demands, dramatically fluctuated channel states, and user departure/arrival. 
The numerical results verify EMCL's effectiveness and fast-response capabilities in adapting to dynamic environments.
\end{itemize} 

The rest of the paper is organized as follows. 
The system model is presented in Section II.
We formulate a resource scheduling problem and develop optimal and suboptimal solutions for performance benchmarks in Section III. 
In Section IV, we model the problem as an MDP and develop an EMCL algorithm.
Numerical results are demonstrated and analyzed in Section V.
Finally, Section VI concludes the paper.

\section{System Model and Problem Formulation}
\subsection{LEO-Terrestrial Network}

\begin{figure}[t] 
\centering
\includegraphics[scale=0.50]{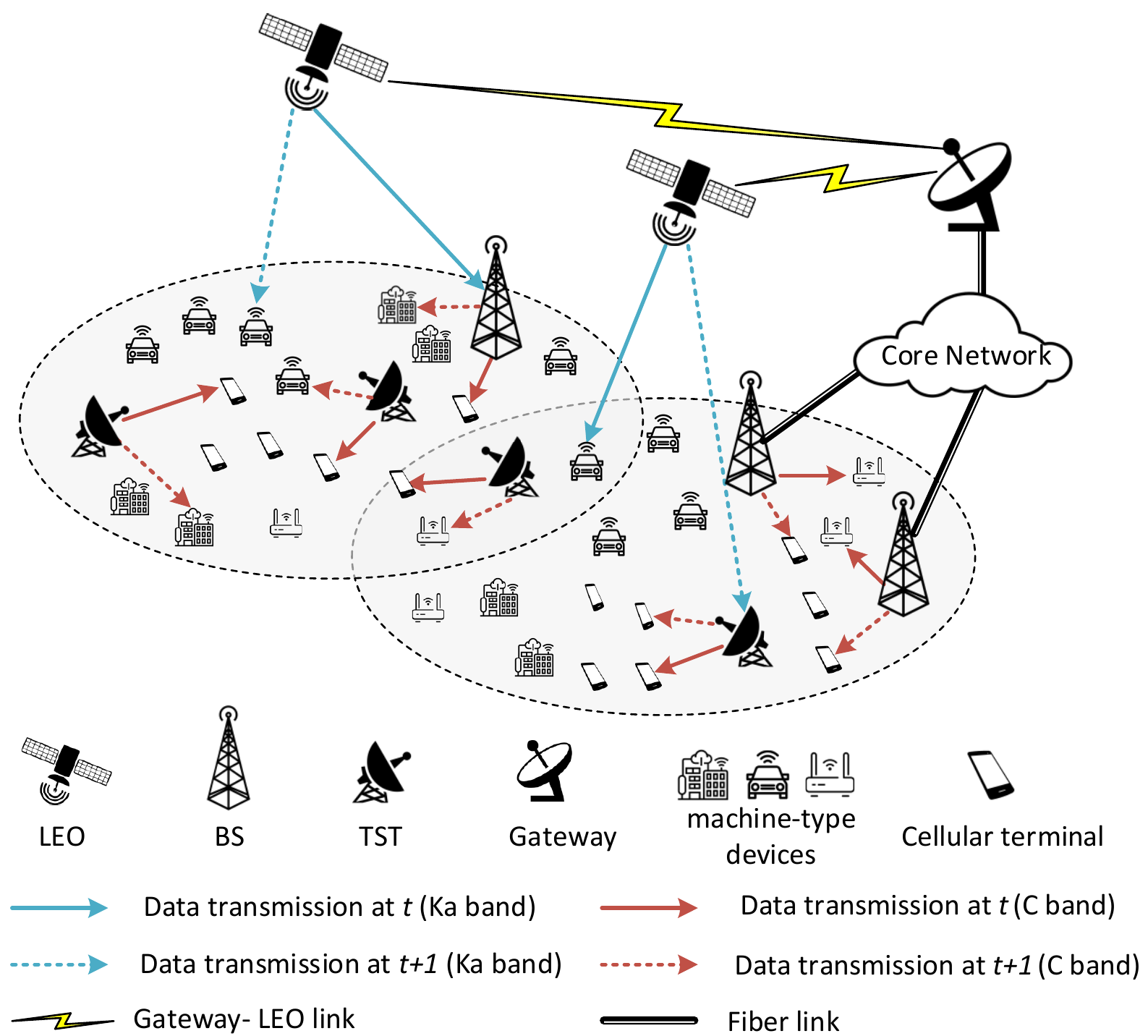}
\caption{An illustrative LEO-terrestrial communication system}
\label{fig:arc}
\vspace{-0.5cm}
\end{figure}

In practice, terrestrial BSs can become over-loaded and congested.
This common issue has received considerable attention from the academia, industry, and standardization bodies, e.g., 3GPP Release 17 \cite{3GPP}.
In this work, we address this challenging issue via developing satellite-aided solutions.
As shown in Fig. \ref{fig:arc}, the BSs with limited resources might not be able to serve all the users and deliver all the requested data demands within a required transmission or queuing delay.
To relieve the burden of the terrestrial BSs, LEO satellites are introduced to offload traffic from BSs or provide backhauling services.
The LEO employs a transparent payload.
The terrestrial communication uses C-band, e.g., sub-6GHz in 5G applications \cite{Mchencband}, whereas the LEO satellite adopts Ka-band to access wider bandwidths since the radio frequency congestion has become a serious issue in the lower bands due to the increased usage of satellites \cite{Oltjon}.
We consider two types of mobile terminals (MTs) that coexist in the system. 
The first type is the normal cellular terminals, e.g., cell phones, that can be served by BSs or terrestrial-satellite terminals (TSTs), typically with high traffic demands. 
The other is the machine-type communication terminals, e.g., massive IoT devices requesting lower data demands, which are equipped with a 3GPP terrestrial-non-terrestrial network (TN-NTN) compliant dual-mode that can be either served by LEO via Ka-band or by BS/TST through C-band \cite{dualmode}.
Compared to conventional cellular BS, TST is a small-size terminal that acts as a flexible and cost-saving access point.
A TST can receive backhauling services from LEO over Ka-band and transmit data to MTs over C-band \cite{DiBLEO2019}.
The terrestrial BSs can request data from the core network through optical fiber links or from the LEO satellites through the BS-LEO link.

We denote $\CK=\{1,...,k,...,K\}$ as the union of ground devices (GDs) served as the receivers, including MTs, BSs and TSTs, and denote $\CN=\{1,...,n,...,N\}$ as the union of transmitters. 
The set $\CN$ can be further divided into two subsets $\CN_1$ and $\CN_2$, where $\CN_1$ consists of the transmitters using Ka-band, i.e., LEO, while $\CN_2$ contains the transmitters occupying C-band, i.e., BSs and TSTs.
We assume that the transmission tasks are delay-sensitive and to be completed within $T$ time slots. 
The time domain is divided by time slots, i.e., $\CT = \{1,...,t,...,T\}$.
In data transmission, each transmitter $n$ serves GD in unicast mode, i.e., no joint transmission and no multi-cast transmission. 
Within a time slot, multiple transmitter-GD links can be activated, forming a link group.
We denote $\CG=\{1,...,g,...,G\}$ as a set by enumerating all the valid link groups.

\subsection{Channel Modeling}
We consider time-varying channels for both satellite and terrestrial communication.
At time slot $t$, the channel state between receiver $k$ and transmitter $n$ can be modeled as:
\begin{align}
h_{k,n,t}= \left\lbrace
\begin{array}{cc}
G^{(T)}_{leo}\cdot G^{(C)}_{k,n,t}\cdot G^{(R)}, & n \in \CN_1, \\
G^{(T)}_{ter}\cdot G^{(C)}_{k,n,t}\cdot G^{(R)}, & n \in \CN_2,
\end{array}
\right.
\end{align}
where $G^{(T)}_{leo}$ and $G^{(T)}_{ter}$ are the transmit antenna gain of LEO and terrestrial BS/TST, respectively.
We assume that all the GDs are equipped with a single receiving antenna, so that their receive antenna gains $G^{(R)}$ are uniform.
$G^{(C)}_{k,n,t}$ represents the channel fading between transmitter $n$ and GD $k$ at time slot $t$.
For LEO-to-GD channel, a widely used channel fading model in \cite {DiBLEO2019, WangSLEO2019, ALSLEO2020} is adopted, which includes free-space path loss, pitch angle fading, atmosphere fading, and Rician small-scale fading:
\begin{align}
\label{eq:GLEOfading}
G^{(C)}_{k,n,t}=\left(\frac{c}{4\pi d_{k,n,t} f_{leo}}\right)^2 \cdot G^{(P)}_{k,n} \cdot A(d_{k,n,t}) \cdot \varphi,
\end{align}
where $c$ is the speed of light, $d_{k,n,t}$ is the propagation distance between LEO and the terminals, $f_{leo}$ is the carrier frequency of LEO, $G^{(P)}_{k,n}$ is the pitch angle fading factor, and $\varphi$ is the Rician fading factor.
$A(d_{k,n,t})$ is the atmospheric loss which is expressed as:
\begin{align}
A(d_{k,n,t}) = 10^{\left( \frac{3\chi d_{k,n,t}}{10H} \right)},
\end{align}
where $\chi$, in $dB/km$, is the attenuation through the clouds and rain, and $H$ is the altitude of LEO.
In downlink transmission, we assume that Doppler shift caused by the high mobility of LEO can be perfectly pre(post)-compensated in the gateway based on the predictable satellite motion and speed \cite{3GPPDoppler}.
For terrestrial channels, i.e., TST/BS-to-MT, $G^{(C)}_{k,n,t}$ consists of the path loss and Rayleigh small-scale fading \cite{GuoKSat2020}, which is given by:
\begin{align}
\label{eq:TEEREfading}
G^{(C)}_{k,n,t}=\left(\frac{c}{4\pi d_{k,n,t} f_{ter}}\right)^2 \cdot \phi,
\end{align}
where $f_{ter}$ is the carrier frequency of TST/BS and $\phi$ is the Rayleigh fading factor.
 
Based on the adopted channel fading models (\ref{eq:GLEOfading}) and (\ref{eq:TEEREfading}), we further model the time-varying channel as the first state Markov channel (FSMC) to capture the time-correlation characteristics and conduct mathematically tractable analysis \cite{FSMC}.
We discretize each channel state $h_{k,n,t}$ into $L$ levels, $\CH=\{h_1,...,h_L\}$, and the transition probability matrix is defined as:
\begin{align}
\mathbf{P}=\left[
\begin{array}{cccc}
P_{1,1}& \cdots& P_{1,L}\\
\vdots& \ddots& \vdots\\
P_{L,1}& \cdots& P_{L,L}
\end{array}
\right],
\end{align}
where an element $P_{l,l'}$ can be written as:
\begin{align}
\label{eq:transitionprob}
P_{l,l'} = \text{Prob}\left[h_{k,n,t+1}\!\!=\!h_{l'}|h_{k,n,t}\!\!=\!h_{l}\right],~h_{l},h_{l'}\in \CH.
\end{align}
That is, at a given time slot $t$, if $h_{k,n,t}=h_l$, $P_{l,l'}$ refers to the probability of channel state at the next time slot $h_{k,n,{t+1}}$ transiting from $h_{l}$ to $h_{l'}$.

\subsection{Optimization Problem}
We formulate a resource scheduling problem for the considered over-loaded LEO-5G systems.
We use binary indicators $\alpha_{k,n,g}$ to represent the activated links in group $g$, where $\alpha_{k,n,g}=1$ if the transmitter-GD link $(n,k)$ is included in group $g$ and will be activated when group $g$ is scheduled, otherwise, 0.
We remark that only valid links, i.e., satisfying the following unicast constraints, can form a link group.
\begin{align}
\label{eq:gpcond01}
&\sum\nolimits_{n\in{\CN}}\alpha_{k,n,g}\leq 1,~\forall~k\in\CK,~g\in\CG,\\
\label{eq:gpcond02}
&\sum\nolimits_{k\in{\CK}}\alpha_{k,n,g}\leq 1,~\forall~n\in\CN,~g\in\CG.
\end{align}
(\ref{eq:gpcond01}) means that each GD $k$ in group $g$ receives data from at most one transmitter, and (\ref{eq:gpcond02}) represents each transmitter $n$ in group $g$ serves no more than one GD.
Confined by (\ref{eq:gpcond01}) and (\ref{eq:gpcond02}), the SINR and the rate of GD $k$ in group $g$ at time slot $t$ are expressed in (\ref{eq:SINR}) and (\ref{eq:RATE}), respectively.
\begin{align}
\label{eq:SINR}
\gamma_{k,g,t} =& \frac{\sum_{n\in\CN_1}h_{k,n,t}\alpha_{k,n,g}p_{k,g}}{\sum_{j\in\CK\setminus k}\sum_{n\in\CN_1}h_{j,n,t}\alpha_{j,n,g}p_{k,g}+\sigma^2}\notag\\
+&\frac{\sum_{n\in\CN_2}h_{k,n,t}\alpha_{k,n,g}p_{k,g}}{\sum_{j\in\CK\setminus k}\sum_{n\in\CN_2}h_{j,n,t}\alpha_{j,n,g}p_{k,g}+\sigma^2},
\end{align}
and
\begin{align}
\label{eq:RATE}
R_{k,g,t} = \Phi B_{k,g} \log_2 (1+\gamma_{k,g,t}),
\end{align}
where $p_{k,g}$ is the transmit power to GD $k$ in group $g$ and $\Phi$ is the duration of each time slot.
We denote $B_{leo}$ and $B_{ter}$ are the fixed bandwidth for LEO and BS/TST, respectively, such that the used bandwidth $B_{k,g}$ for GD $k$ in group $g$ can be calculated by $B_{leo}\sum_{n\in\CN_1}\alpha_{k,n,g} + B_{ter}\sum_{n\in\CN_2}\alpha_{k,n,g}$.
We define the decision variables as $\mathbf{x}=[x_{1,1},...,x_{g,t},...,x_{G,T}]$ where
\begin{align}
x_{g,t}=&\left\lbrace
\begin{array}{ll}
1, & \text{if group } g \text{ is scheduled at time slot } t, \\
0, & \text{otherwise}.
\end{array}
\right.\notag
\end{align}
In the objective design, considering the over-loaded scenario with densely deployed users, the system with limited resources may not be able to satisfy every user's data demand $D_k$ (bits). 
We then consider a composite utility function in (\ref{defy:1}), and define that GD $k$ is served, i.e., $f_k(\mathbf{x})=1$, when a threshold $D_k' (D_k'<D_k)$ is satisfied. 
\begin{align}
\label{defy:1}
f_k(\mathbf{x}) =& \mathds{1}\left(\sum_{t\in\CT}\sum_{g\in\CG}R_{k,g,t}x_{g,t}-D'_k\right),
\end{align}
where $\mathds{1}(\cdot)$ is an indicator function such that 
$\mathds{1}(\beta)=\left\lbrace
\begin{array}{cc}
1,&\text{if~} \beta>0\\
0,&\text{if~} \beta\leq 0\\
\end{array}
\right.$.
We convert the non-linear function $f_k(\mathbf{x})$ to a linear function by introducing auxiliary variables $\mathbf{y}=[y_1,...,y_k,...,y_K]$ and linear constrains (\ref{eq:p12}), where $y_k=f_k(\mathbf{x})$.
The optimization problem is formulated as:
\begin{subequations}
\label{eq:p1}
\begin{align}
\label{eq:p1obj}
\qquad {\bf P1}:
&\min_{x_{g,t},y_{k}}~ f(\textbf{x}, \textbf{y}) = \eta_0  \left(\sum_{k\in\CK}y_{k} - K\right)^2  + \notag\\
& \sum_{k \in \CK}\eta_{k}\left(\sum_{t \in \CT}\sum_{g \in \CG} R_{k,g,t}x_{g,t} - D_{k}\right)^2 &  \\
s.t.~&  
\label{eq:p10}
\bar{\gamma}_{k}\,-\,\gamma_{k,g,t}\leq V\left(1-x_{g,t}\sum_{n\in\CN}\alpha_{k,n,g}\right), \notag\\
\qquad & 
~\forall k\in \CK,g\in \CG,t\in \CT, &\\
\label{eq:p11}
\qquad &  
\sum_{g\in\CG} x_{g,t} \leq 1,~\forall t\in \CT,   &\\
\label{eq:p12}
\qquad &  
D'_ky_{k}\leq \sum_{t\in\CT}\sum_{g\in\CG}R_{k,g,t}x_{g,t},~\forall k\in \CK, &\\
\label{eq:p13}
\qquad &  
x_{g,t}\in \{0,1\},~\forall g \in \CG,~t\in \CT,  &\\
\label{eq:p14}
\qquad &  
y_{k}\in \{0,1\},~\forall ~k\in \CK,  &
\end{align} 
\end{subequations}
where $\bar{\gamma}_k$ is the SINR threshold of GD $k$, $V$ is a positive sufficiently large value, and $\eta_0,...,\eta_K$ are the weight factors.
The GDs with higher priority will be assigned with larger $\eta_k$, e.g., VIP clients subscribing advanced services or the users requesting emergency communications.
The objective (\ref{eq:p1obj}) aims at serving as many GDs as possible and delivering more data to GDs such that the two gaps in the over-loaded system can be minimized.
The first term in the objective is used to keep the user fairness.
The optimal scheduler is encouraged to meet the minimum requirement $D_k'$ to serve more users and gain utilities from the first term.
In the second term, when some users' weights $\eta_k$ are considerable, consuming some resources and satisfying higher demand $D_k$ can be preferred. 
\begin{itemize}
\item The constraints (\ref{eq:p10}) mean that if GD $k$ in group $g$ is scheduled at time slot $t$, i.e., $x_{g,t}\sum_{n\in\CN}\alpha_{k,n,g}=1$, the SINR of GD $k$ should be higher than the threshold $\bar{\gamma}_{k}$.
\item The constraints (\ref{eq:p11}) represent no more than one group can be scheduled in a time slot.  
\item In constraints (\ref{eq:p12}), we define that if GD $k$ is served, i.e., $y_k=1$, the received data should be larger than $D'_k$. 
\end{itemize}

\section{Characterization on Solution Development}
In this section, we propose an optimal method and a heuristic approach as the offline benchmarks for small-medium and large-scale instances, respectively.
In addition, we outline conventional online-learning solutions and their limitations.

\subsection{The Proposed Optimal and Sub-optimal Solutions}
Towards the optimum of ${\bf P1}$, we first identify the convexity of ${\bf P1}$ when the binary variables are relaxed. 
\begin{lemma}
\label{le:convex}
The relaxation problem of ${\bf P1}$ is convex.
\end{lemma}
\begin{proof}
See Appendix \ref{app:lemma1}.
\end{proof}

Based on Lemma \ref{le:convex}, we conclude that ${\bf P1}$ is an integer convex optimization problem.
The optimum can be obtained by B\&B that solves a convex relaxation problem at each node \cite{bbbook}.
Although the complexity increases exponentially, the B\&B based approach can provide a performance benchmark at least for small-medium instances.

To reduce the complexity in solving the large-scale problems, we develop a suboptimal algorithm.
We observe that ${\bf P1}$ has a variable-splitting structure, which motivates the development of ADMM based approaches \cite{ADMM}.
The algorithm is summarized in Alg. \ref{alg:Alg1}, first solving the convex relaxation problem of ${\bf P1}$ based on ADMM (in lines 2-8), followed by a rounding operation (in lines 9-13).
In ADMM, we divide the relaxed variables into $T+1$ blocks $\hat{\mathbf{x}}_1,...,\hat{\mathbf{x}}_T,\hat{\mathbf{y}}$, where $\hat{\mathbf{x}}_t = [\hat{x}_{1,t},...,\hat{x}_{G,t}]$, and introduce auxiliary variables $\mathbf{z}=[z_1,...,z_K]$, where
\begin{align}
z_k = D'_k\hat{y}_k-\sum_{g\in\CG}\sum_{t\in\CT}R_{k,g,t}\hat{x}_{g,t},&\forall k\in\CK.
\end{align}
The inequality constraints (\ref{eq:p12}) are replaced by:
\begin{align}
z_k \leq 0,&~~\forall k\in\CK.
\end{align}
The augmented Lagrangian function is expressed as:
\begin{align}
\label{eq:augmented}
&L(\hat{\mathbf{x}}_1,...,\hat{\mathbf{x}}_T,\hat{\mathbf{y}},\mathbf{z},\boldsymbol{\lambda}) \notag\\
= & f(\hat{\mathbf{x}}, \hat{\mathbf{y}}) + \sum_{k\in\CK}\lambda_k\left(z_k-D'_k\hat{y}_k+\sum_{g\in\CG}\sum_{t\in\CT}R_{k,g,t}\hat{x}_{g,t}\right) \notag \\
+ &\frac{\rho}{2}\sum_{k\in\CK}\|z_k-D'_k\hat{y}_k+\sum_{g\in\CG}\sum_{t\in\CT}R_{k,g,t}\hat{x}_{g,t}\|^2,
\end{align}
where $\rho>0$ is the penalty parameter and $\boldsymbol{\lambda} = [\lambda_1,...,\lambda_K]$ are the lagrangian multipliers.
We define $I_{iter}$ as the total number of iterations of the algorithm.
In each iteration $i$, ADMM update each variable block as follows (in line 5) and update multipliers (in line 6):
\begin{align}
&\hat{\mathbf{x}}_t^{i+1} = \argmin\limits_{\hat{\mathbf{x}}_t\in\mathcal{X}_t}L(\hat{\mathbf{x}}_1^{i},...,\hat{\mathbf{x}}_T^{i},\hat{\mathbf{y}}^{i},\mathbf{z}^{i},\boldsymbol{\lambda}^{i}),~\forall t\in \CT,\label{eq:updatexhat}\\
&\hat{\mathbf{y}}^{i+1} = \argmin\limits_{\hat{\mathbf{y}}\in\mathcal{Y}}L(\hat{\mathbf{x}}_1^{i},...,\hat{\mathbf{x}}_T^{i},\hat{\mathbf{y}}^{i},\mathbf{z}^{i},\boldsymbol{\lambda}^{i}),\label{eq:updateyhat}\\
&\mathbf{z}^{i+1} = \argmin\limits_{\mathbf{z}\in\mathcal{Z}}L(\hat{\mathbf{x}}_1^{i},...,\hat{\mathbf{x}}_T^{i},\hat{\mathbf{y}}^{i},\mathbf{z}^{i},\boldsymbol{\lambda}^{i}),\label{eq:updatezhat}
\end{align}
where
$\mathcal{X}_t = \{\mathbf{x}_t|(\text{\ref{eq:p10}), (\ref{eq:p11}), (\ref{eq:p12})}\}$, 
$\mathcal{Y} = \{\mathbf{y}|0 \leq y_k \leq 1\}$ and 
$\mathcal{Z} = \{\mathbf{z}|z_k \leq 0\}$.
When ADMM terminates, the continuous solution $\hat{x}_{g,t}$ is obtained in line 8. 
The rounding process is then carried out in lines 10-13 to convert the largest $\hat{x}_{g,t}$ in each time slot to 1 (selecting the most promising group $g$ for each $t$) and keep others 0. 

\begin{algorithm}[t]
  \caption{ADMM-HEU}
  \label{alg:Alg1}
  \begin{algorithmic}[1]
  	\STATE \textbf{input:}
  	$D_k$, $D'_k$ and $R_{k,n,t}$.
  	\STATE Relax ${\bf P1}$ to a continuous problem ${\bf P1}'$.
  	\STATE Initialize $\hat{\mathbf{x}}_t^{0}$, $\hat{\mathbf{y}}^{0}$, $\mathbf{z}^{0}$, $\boldsymbol{\lambda}^{0}$ and $i=0$.
	\FOR {$i=0,...,I_{iter}$}
	\STATE Update $\hat{\mathbf{x}}_t, \hat{\mathbf{y}}$ and $\mathbf{z}$ by Eq. (\ref{eq:updatexhat}), (\ref{eq:updateyhat}) and (\ref{eq:updatezhat}).
	\STATE 
	$\lambda_k^{i+1} = \lambda_k^i + \rho\left(z_k^i-D'_k\hat{y}_k^i+\sum\limits_{g\in\CG}\sum\limits_{t\in\CT}R_{k,g,t}\hat{x}_{g,t}^i\right)$.
	\ENDFOR
	\STATE Obtain relaxed solution $\hat{x}_{g,t}$.
	\FOR {$t \in \CT$}
		\STATE Find $g^{\dagger} = \argmax\limits_{g\in\CG}\{\hat{x}_{1,t},...,\hat{x}_{G,t}\}$.
		\STATE Set $x_{g^{\dagger},t}^* = 1$ and $x_{g,t}^* = 0,~\forall g \neq g^{\dagger}$.
   \ENDFOR
	\STATE Calculate $y_k^*$ based on Eq. (\ref{defy:1}).
	\STATE \textbf{output:}
	$x_{g,t}^*$ and $y_k^*$

\end{algorithmic}
\end{algorithm}

The developed ADMM-HEU can provide sub-optimal benchmarks within an acceptable time span, since the subproblems in (\ref{eq:updatexhat})-(\ref{eq:updatezhat}) can be solved in a parallel manner and with a smaller size than the original problem.
However, ADMM-HEU is still an offline algorithm because iteratively solving the subproblems in (\ref{eq:updatexhat})-(\ref{eq:updatezhat}) requires a considerable amount of time, which might not sufficient for fast adaptation to network variations.

\subsection{Conventional Online-Learning Solutions and Limitations}

To enable an intelligent and online solution, we address the problem from an RL perspective.
Firstly, we briefly introduce actor-critic and meta-critic learning approaches as a basis to present the proposed EMCL. 
AC is an RL algorithm that takes advantage of both value-based methods, e.g., Q-learning, and policy-based methods, e.g., REINFORCE, with fast convergent properties and the capability to deal with continuous action spaces \cite{IntroRL}.
The learning agent in AC contains two components, where the actor is responsible for making decisions while the critic is used for evaluating the decisions by the value functions.
Specifically, at each learning step $t$\footnote{In this paper, a learning step corresponds to a time slot.}, the actor takes action based on a stochastic policy, i.e., $a_t \sim \pi(a|s_t)$, where $\pi(a|s_t)$ is the probability of taking an action under state $s_t$, typically following the Gaussian distribution \cite{weiuserschedule}.
The critic is to generate a Q-value function $Q(s_t,a_t)=\mathbb{E}_{\pi}[\bar{r}_t|s_t,a_t]$, where $\bar{r}_t$ is the accumulated reward at step $t$, and $\mathbb{E}_{\pi}[\beta]$ is the expected value of $\beta$ over the policy $\pi$. 
The goal of the learning agent is to find a policy to maximize the expected accumulated reward (or Q-value).

\begin{figure}[h]
\centering
\includegraphics[scale=0.45]{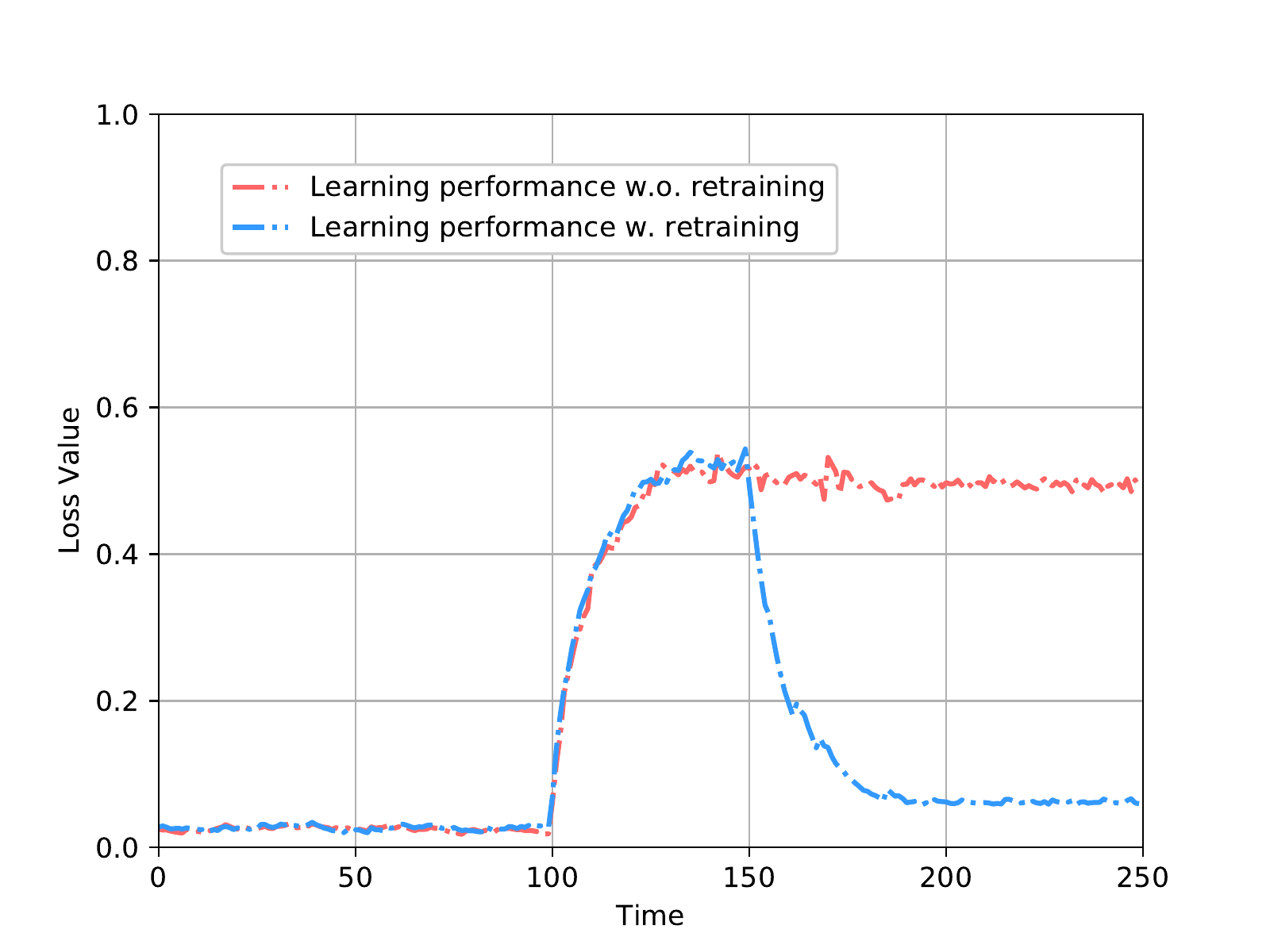}
\caption{Evolution of loss over time-varying demands.}
\label{fig:label04}
\end{figure} 

A critical issue in conventional learning approaches, including AC, is that the performance of a learning model largely depends on the adopted training or observed data sets. 
The practical LEO-5G systems are highly complex and dynamic, such as fast and dramatic variations in channel states, user demands, user arrival/departure, and network topologies.
This typically causes the new inputs to become no longer relevant to the statistical properties of the historical data \cite{conceptdrift}.
As a consequence, the previous learning model can become invalid and the model may need to be re-trained to adapt.
To illustrate this impact, we use Fig. \ref{fig:label04}, as an example, to depict a typical evolution of AC's loss value over time-varying demands.
From 0 to 100 time slots, the demand is time-varying but follows historical statistical properties, e.g., fluctuating within a certain range or following a certain distribution, leading to a well-adapted AC with low and stable loss values.
When a surged demand is generated at the 100-th time slot, the new input deviates the statistics.
The AC model becomes inapplicable to the new environment, evidenced by the rapidly deteriorating loss values.
When the agent in AC consumes a considerable amount of time in new data collection and re-training, the performance can return to the previous level.

To address this issue, meta-critic-based approaches become an emerging technique that takes advantage of a variety of previously observed tasks to infer the meta-knowledge, such that a new learning task can be quickly trained with few observations \cite{FinnMAML}.  
Meta-critic learning combines meta-learning with an AC framework to enhance the generalization ability.
However, conventional meta-critic learning is not effective in dealing with the large discrete space in ${\bf P1}$.
In addition, there is no uniform standard to parameterize the learning model and extract meta-knowledge in dynamic environments.
Thus, we propose an EMCL algorithm to enable an efficient dynamic-adaptive solution.

\section{The Proposed EMCL Algorithm}
In this section, we elaborate the proposed EMCL algorithm, firstly staring from outlining the EMCL framework, then detailing the tailored design.

\subsection{EMCL Framework}

\subsubsection{MDP Reformulation}
First, we reformulate the original problem ${\bf P1}$ as an MDP by defining action, state and reward.
\begin{itemize}
\item As the actor is to select a group from set $\CG$ at each time slot $t$, the action is defined as an assigned link group,
\begin{align}
a_t = g \in \CG.
\end{align}
\item The state consists of the channel coefficients $h_{k,n,t}$, modeled as FSMC with the
transition probability defined in (\ref{eq:transitionprob}), and the delivered data for user $k$ up to time slot $t$, where $b_{k,t} = b_{k,t-1} + R_{k,a_t,t}$.
\begin{align}
s_t=\{h_{1,1,t},...,h_{K,N,t}, b_{1,t},...,b_{K,t}\}.
\end{align}
All possible states are included in the state space $\CS$.
The next state only depends on the current state and action but is irrelevant to the past, which means the state transition from $s_t$ to $s_{s+1}$ follows the Markov property \cite{IntroRL}.
\item The reward is closely related to the objective of ${\bf P1}$.
We define the reward as (\ref{eq:designreward}).
\begin{align}
\label{eq:designreward}
r_t = \sum_{k=0}^{K} \eta_k (\Delta_{k,t-1}^2-\Delta_{k,t}^2),
\end{align}
where $\Delta_{k,t}=
\left\lbrace
\begin{array}{ll}
\sum\limits_{k=1}^K \mathds{1}\left(b_{k,t}-D'_k\right)-K, & k=0,\\
b_{k,t}-D_k, & k\neq 0.
\end{array}
\right.$
Then, the accumulated reward at step $t$ is given by $\bar{r}_t = \sum_{t'=t}^{T}\gamma^{t'-t}r_{t'}$, where $\gamma \in [0,1]$ is a discounted factor.
\end{itemize}
Under the designed MDP, we verify the consistency between the goals of the RL algorithm and the original optimization problem such that the policy provided by the learning agent can minimize the objective in ${\bf P1}$.

\begin{lemma}
\label{le:2} 
When $\gamma = 1$, the objective of the learning agent is equivalent to that of the optimization problem ${\bf P1}$.
\end{lemma}
\begin{proof}
See Appendix \ref{app:lemma2}
\end{proof}

\subsubsection{Meta Critic and Task-Specific Actor}
As shown in Fig. \ref{fig:mc_structure}, we design a hierarchical structure in EMCL containing a meta critic\footnote{In this paper, ``meta-critic learning'' refers to an algorithm that combines AC and meta-learning while ``meta critic'' refers to the critic in the framework.} and multiple actors.
The learning model is trained over multiple tasks, $\CJ^{(1)},...,\CJ^{(I)}$. 
The meta critic can evaluate the actions for any task with a Q-value while each actor provides a stochastic policy dedicated to a specific task.
At time step $t$, $s_t^{(i)}$, $a_t^{(i)}$, and $r_t^{(i)}$ represent the state, action, and reward for task $i$, respectively.
An episode $\CD^{(i)}=\{s_1^{(i)},a_1^{(i)},r_1^{(i)}...,s_T^{(i)},a_T^{(i)},r_T^{(i)}\}$ can be sampled from the first step to the terminal step $T$.
We denote $\CD^{(i)}_{[u, w]}$ as a segment of $\CD^{(i)}$ from step $u$ to $w$, i.e., $\CD^{(i)}_{[u,w]}=\{s_{u}^{(i)},a_{u}^{(i)},r_{u}^{(i)},...,s_{w}^{(i)},a_{w}^{(i)},r_{w}^{(i)}\}$.
\begin{figure}[t]
\centering
\includegraphics[scale=0.58]{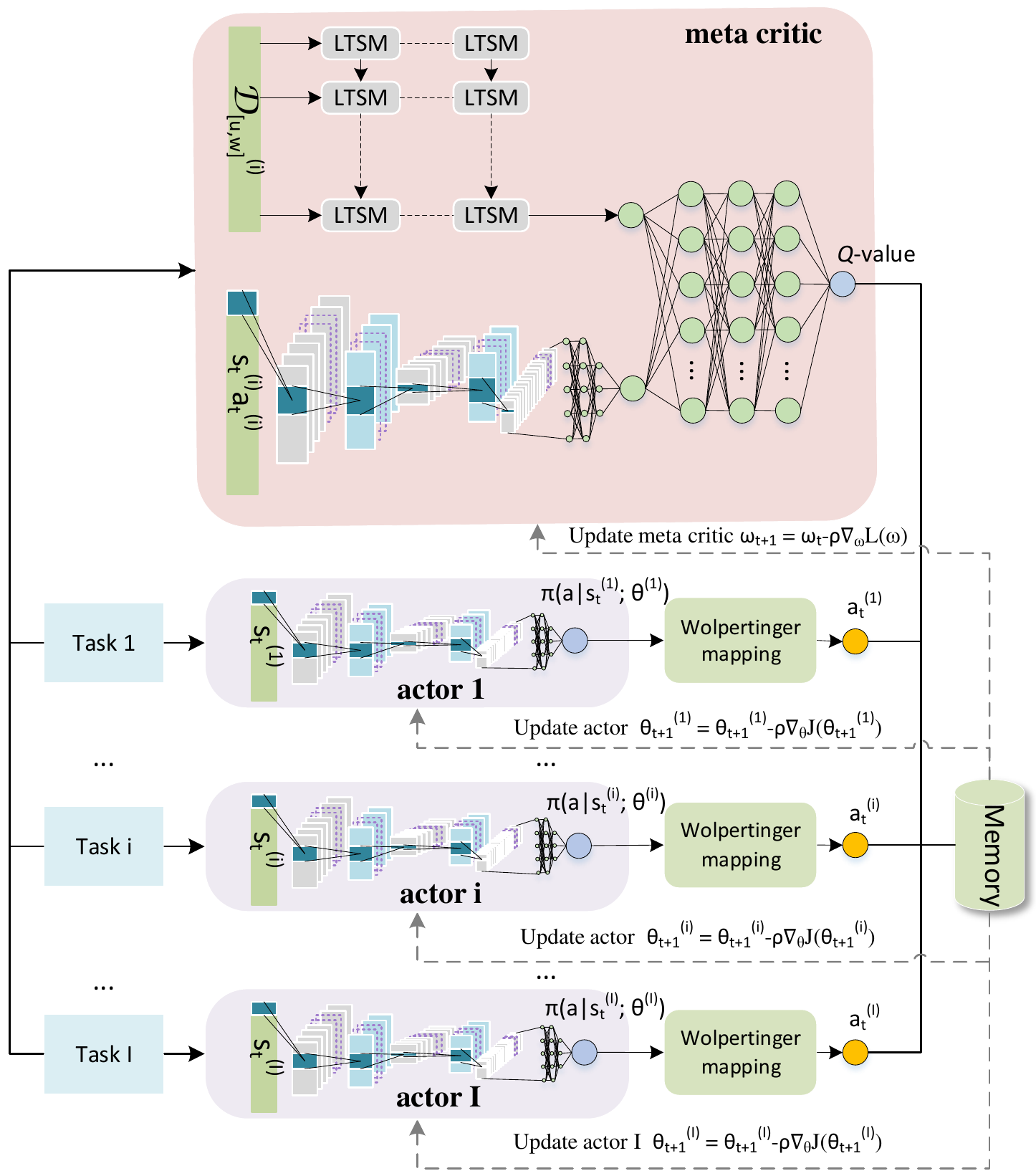}
\caption{The proposed EMCL framework.}
\label{fig:mc_structure}
\end{figure} 
Since the explicit meta critic and actors are difficult to obtain, we adopt the function approximation method.
The meta critic is parameterized as a neural network (NN) with the weights $\boldsymbol{\omega}$, i.e., $Q(s_t^{(i)}, a_t^{(i)}, \CD_{[t-\bar{t},t-1]}^{(i)};\boldsymbol{\omega})$.
We note that, in addition to $s_t^{(i)}$ and $a_t^{(i)}$, the input includes the most recent $\bar{t}$ samples $\CD_{[t-\bar{t},t-1]}^{(i)}$.
Each task-specific actor is modeled as an NN $\pi(a|s_t^{(i)};\boldsymbol{\theta}^{(i)})$ with the weights $\boldsymbol{\theta}^{(i)}$.

To optimize the weights, we minimize the loss functions by gradient descend.
The loss function of the meta critic $L(\boldsymbol{\omega})$ is defined as the average temporal difference (TD) error over all tasks:
\begin{align}
\label{eq:lossmetacritic}
L(\boldsymbol{\omega})=&\frac{1}{I}\sum_{i=1}^{I}\mathbb{E}_{\pi(\boldsymbol{\theta}^{(i)})}\left[(Q(s_{t+1}^{(i)}, a_{t+1}^{(i)}, \CD^{(i)}_{[t-\bar{t}+1,t]}; \boldsymbol{\omega})-r_t\right.\notag\\
&\left.-\gamma Q(s_t^{(i)}, a_t^{(i)}, \CD^{(i)}_{[t-\bar{t},t-1]}; \boldsymbol{\omega})\right]^2,
\end{align}
where the TD error reflects the similarity between the estimated Q-value and actual Q-value.
For the task-specific actor, the loss function $J(\boldsymbol{\theta}^{(i)})$ is the negative Q-value:
\begin{align}
J(\boldsymbol{\theta}^{(i)})=\mathbb{E}_{\pi(\boldsymbol{\theta}^{(i)})}\left[-Q(s_t^{(i)},a_t^{(i)},\CD^{(i)}_{[t-\bar{t},t-1]};\boldsymbol{\omega})\right],
\end{align} 
such that minimizing $J(\boldsymbol{\theta}^{(i)})$ is equivalent to maximizing the expected accumulated reward.
The update rules are given by:
\begin{align}
\label{eq:upmetacritic}
\boldsymbol{\omega}_{t+1} =& \boldsymbol{\omega}_{t} - \rho \nabla_{\boldsymbol{\omega}}L(\boldsymbol{\omega}),\\
\label{eq:upactor}
\boldsymbol{\theta}_{t+1}^{(i)}=&\boldsymbol{\theta}_{t}^{(i)}-\rho \nabla_{\boldsymbol{\theta}^{(i)}}J(\boldsymbol{\theta}^{(i)}).
\end{align}
Based on the fundamental results of the policy gradient theorem \cite{IntroRL}, the gradients of $L(\boldsymbol{\omega})$ and $J(\boldsymbol{\theta}^{(i)})$ are:
\begin{align}
\label{eq:gradmetacritic}
&\nabla_{\boldsymbol{\omega}}L(\boldsymbol{\omega}) = \frac{1}{I}\sum_{i=1}^{I}\left[2L(\boldsymbol{\omega})\nabla_{\boldsymbol{\omega}}(Q(s_{t+1}^{(i)},a_{t+1}^{(i)},\CD^{(i)}_{[t-\bar{t}+1,t]};\boldsymbol{\omega})\right.\notag\\
&\left.-Q(s_t^{(i)},a_t^{(i)},\CD^{(i)}_{[t-\bar{t},t-1]};\boldsymbol{\omega}))\right],\\
\label{eq:gradactor}
&\nabla_{\boldsymbol{\theta}^{(i)}}J(\boldsymbol{\theta}^{(i)})= -Q(s_t^{(i)},a_t^{(i)},\CD;\boldsymbol{\omega})\nabla_{\boldsymbol{\theta}^{(i)}}\log \pi(a|s_t^{(i)};\boldsymbol{\theta}^{(i)}).
\end{align}

\subsubsection{Algorithm Summary}

\begin{algorithm}[t]
  \caption{EMCL}
  \label{alg:Alg2}
  \begin{algorithmic}[1]
  \REQUIRE 
  \STATE \textbf{input}: Multiple task samples; initial $\boldsymbol{\omega}_0$.
	\FOR {each learning episode}
		\STATE Sample $I$ tasks and initialize $\boldsymbol{\omega}_0, \boldsymbol{\theta}^{(1)}_0,...,\boldsymbol{\theta}^{(I)}_0$.
		\FOR {each learning step $t$}
			\FOR {each task $i$}
			\STATE Obtain Q-value by the meta critic in (\ref{eq:mc_Q}).
			\STATE Obtain stochastic policy by the actor in (\ref{eq:sa_pi}).
			\STATE Take actions $a^{(i)}_t$ by the Wolpertinger approach.
			\STATE Store tuples $\{s_t^{(i)},s_{t+1}^{(i)},a_t^{(i)},r_t^{(i)}\}$ in the memory.
			\STATE Take a batch of data and update $\boldsymbol{\theta}^{(i)}$ by (\ref{eq:upactor}). 
	\ENDFOR
	\STATE Update $\boldsymbol{\omega}$ by (\ref{eq:upmetacritic}). 
	\ENDFOR
	\ENDFOR
	\STATE \textbf{output}: The well-trained meta critic $\boldsymbol{\omega}^*$.
	
	\ENSURE 
	\STATE \textbf{input}: A new task; initial $\boldsymbol{\theta}_0$; well-trained meta critic $\boldsymbol{\omega}^*$.
  \FOR {each learning episode}
		\FOR {each learning step $t$}
			\STATE Obtain Q-value by the meta critic in (\ref{eq:mc_Q}).
			\STATE Obtain stochastic policy by the actor in (\ref{eq:sa_pi}).
			\STATE Take an action $a_t$ by the Wolpertinger approach.
			\STATE Store tuples $\{s_t,s_{t+1},a_t,r_t\}$ in the memory.
			\STATE Take a batch of data and update $\boldsymbol{\theta}$ by (\ref{eq:upactor}). 
	\ENDFOR
	\ENDFOR
	\STATE \textbf{output}: The optimal actor $\boldsymbol{\theta}^*$.
  \end{algorithmic}
\end{algorithm}

We summarize the proposed EMCL in Alg. \ref{alg:Alg2}, which includes two phases: the meta training phase and the online learning phase.
For the former, the learning model is trained over different learning tasks.
At each learning episode, we sample $I$ learning tasks.
We obtain the approximated Q-value (in line 6) and stochastic policy (in line 7) by the approximation functions.
The final actions are determined by the Wolpertinger approach in line 8, which will be elaborated in the following subsection.
In line 9, the memory is used to store the experienced learning tuples $\{s_t^{(i)},s_{t+1}^{(i)},a_t^{(i)},r_t^{(i)}\}$.
At each step, we extract a batch of tuples from the memory as the training data for updating $\boldsymbol{\omega}$ and $\boldsymbol{\theta}^{(i)}$ by (\ref{eq:upmetacritic}) and (\ref{eq:upactor}) in line 10 and 12, respectively.
In the online learning phase, given a new task, the well-trained meta critic $\boldsymbol{\omega}^*$ can be directly used to estimate the Q-value and only the actor needs to be re-trained.

\subsection{Tailored Designs in EMCL}

\subsubsection{Parameterization with Hybrid Neural Networks}

There is no uniform standard for parameterization in conventional meta-critic learning.
Considering dynamic environments, the distribution of the new input data and the previous observations may deviate. 
Towards fast adaptation to the dynamic environment, the critic should be able to identify different tasks, where the information for task identification can be refined from the experienced data, which usually forms time-related series \cite{Sungmetacritic}.
The widely used DNN might have limitations in efficiency and in mining features from time-series data due to massive number of weights and feed-forward structure.
In the proposed EMCL, we design tailored neural networks to enable the meta critic and the actors to fit the complex nonlinear relationships and extract the meta-knowledge from historical data.

As shown in Fig. \ref{fig:mc_structure}, for the meta critic, a hybrid neural network (HNN) combing convolutional neural network (CNN), long-short term memory (LSTM), and DNN is applied to learn the features from the current state-action pairs and historical trajectories \cite{goodfellow}.
Thereinto, CNN is computation-efficient via adopting the parameter sharing and pooling operations, and is effective to extract spatial features from the input data.
LSTM, as a type of recurrent neural network, has advantages in extracting features from time-related sequential data.
Thus, in the designed meta critic, the CNN is used to extract a feature from the current action-state pair $s_t^{(i)},a_t^{(i)}$ to evaluate the decisions made by the actor.
The LSTM is adopted for task-identification based on the time-series data $\CD_{[t-\bar{t},t-1]}^{(i)}$.
We denote $f_{cnn}(\mathbf{x};\mathbf{w})$, $f_{lstm}(\mathbf{x};\mathbf{w})$ and $f_{dnn}(\mathbf{x};\mathbf{w})$ as the outputs of CNN, LSTM, and DNN, respectively, which are the functions of input $\mathbf{x}$ and weight $\mathbf{w}$.
The features output from CNN and LSTM are:
\begin{align}
\xi_1 =& f_{cnn}(s_t^{(i)},a_t^{(i)};\boldsymbol{\omega}_{cnn}),\\
\xi_2 =& f_{lstm}(\CD_{[t-\bar{t},t-1]}^{(i)};\boldsymbol{\omega}_{lstm}).
\end{align}
Then, we take the features as input and pass them through a fully-connected DNN to obtain the approximated Q-value:
\begin{align}
\label{eq:mc_Q}
Q^{\pi}(s_t^{(i)},a_t^{(i)},\CD_{[t-\bar{t},t-1]}^{(i)};\boldsymbol{\omega}) = f_{dnn}(\xi_1,\xi_2;\boldsymbol{\omega}_{dnn}).
\end{align}
For the task-specific actors, we adopt CNN as the approximator which takes the current state as the input and outputs the mean $\mu$ and variance $\vartheta^2$ of the stochastic policy.
We assume the stochastic policy follows Gaussian distribution $N(\mu,\vartheta^2)$, such that
\begin{align}
&[\mu, \vartheta^2] =  f_{cnn}(s_t^{(i)};\boldsymbol{\theta}^{(i)}),\\
\label{eq:sa_pi}
&\pi(a|s_t^{(i)};\boldsymbol{\theta}^{(i)}) =N(\mu,\vartheta^2).
\end{align}
 
\subsubsection{Action Mapping with the  Wolpertinger Policy}
The decision variables in $\bf P1$ are discrete such that we need to map the action from the stochastic policy to a discrete action space.
However, the previous action mapping policies in meta-critic learning are not efficient since the action space is large for $\bf P1$.
Thus, in EMCL, the Wolpertinger policy is adopted for faster convergence \cite{Wolper}.

Following the stochastic policy $\pi$, the actor first produces an action $\hat{a}$ with continuous value, i.e.,
\begin{align}
f_{\pi}:\CS\rightarrow\hat{\CA},~~f_{\pi}(s)=\hat{a},
\end{align}
where $f_{\pi}$ is a mapping from the state space $\CS$ to a continuous action space $\hat{\CA}$ under the policy $\pi$.
As the real action space $\CG$ is discrete in $\bf P1$, the following two conventional approaches can be used for discretization \cite{IntroRL}:
\begin{itemize}
\item $\text{Simple approach:}~~a^*_{s}=\argmin\limits_{a\in\CG}|a-\hat{a}|^2.$
\item $\text{Greedy approach:}~~a^*_{g}=\argmax\limits_{a\in\CG}Q(s,a).$
\end{itemize}
The simple approach is to select the closest integer value to $\hat{a}$.
This approach may result in a high probability of deviating from the optimum, especially at the beginning of learning, and further lead to slow convergence \cite{IntroRL}.
The greedy approach optimizes Q-value at each step but the complexity is proportional to the exponentially increasing space $\CG$ \cite{IntroRL}.
To achieve a trade-off between the complexity and learning performance, the Wolpertinger mapping approach is considered.
\begin{itemize}
\item $\text{Wolpertinger approach:}~~a^*_{w}=\argmax\limits_{a\in\CM^*}Q(s,a),$
\end{itemize}
where $\CM^*$ is a subset of $\CG$ and contains $M$ nearest neighbors of $\hat{a}$.
In the Wolpertinger approach, the final action is determined by selecting the highest-scoring action from $\CM^*$.
The Wolpertinger mapping becomes the greedy approach and simple approach when $M=|\CG|$ and $M=1$, respectively, and the solution of simple approach $a^*_{s}$ is included in $\CM^*$. 
\begin{lemma}
\label{le:3}
We assume $\CM^*=\{a_1,...,a_M\}$ and
$\left\lbrace
\begin{array}{ll}
Q(s, a_m)\sim U(Q(s, a^*_{s})-\kappa, Q(s, a^*_{s})+\kappa),&~m\neq m'\\
Q(s, a_m) = Q(s, a^*_{s}),&~m=m',
\end{array}
\right.
$
where $U(a,b)$ refers to uniform distribution and $\kappa$ is a constant, 
then
\begin{align}
\mathbb{E}\left[Q(s,a^*_{w})\right]=Q(s,a^*_{s}) + \kappa\left(1-\frac{2(2^M-1)}{M\cdot 2^M}\right).
\end{align}
\end{lemma}
\begin{proof}
See Appendix \ref{app:lemma3}
\end{proof}
From Lemma \ref{le:3}, when $M>1$, $\mathbb{E}\left[Q(s,a^*_{w})\right]>Q(s,a^*_{s})$, which means that the Wolpertinger approach finds the actions with higher Q-values than the simple approach at each learning step, and a larger $M$ leads to a higher expected Q-value.
In addition, the complexity of the Wolpertinger approach is lower than the greedy approach as the size of searching space decreases from $|\CG|$ to $|\CM^*|$.
Thus, for the problems with huge discrete spaces, the Wolpertinger approach enables fast convergence to the maximum Q-value with a proper $M$.

\subsection{Complexity Analysis for EMCL}

For the meta critic, an HNN, composed of CNN, LSTM, and DNN, is employed to estimate the Q-value.
We assume CNN includes $V_1$ convolutional layers.
We denote $o_{c,v}$, $o_{k,v}$, $o_{f,v}$ are the number of convolutional kernels, the spatial size of the kernel, and the spatial size of the output feature map in the $v$-th layer, respectively.
The stripe of kernel is 1, and the input size is $o_{c,0} = K(N+1)+1$.
The time complexity of CNN is $\CO\left(\sum_{v=1}^{V_1}o_{c,v-1} \varrho_v \right)$, where $\varrho_v=o_{k,v}^2o_{c,v}o_{f,v}^2$ \cite{CNNcomplexity}.
For the LSTM, we consider $V_{2}$ layers, and denote $o_{l,v}$ and $o_{e,v}$ are the input size and number of memory cells for layer $v$, respectively, where $o_{l,0} = m(K(N+1)+1)$.
The time complexity is given by $\CO\left(\sum_{v=1}^{V_2} o_{e,v}(4o_{l,v-1} + \varsigma_v)\right)$, where $ \varsigma_v = 4o_{e,v}+o_{l,v}+3$ \cite{LSTMcomplexity}.
For the fully-connected DNN, the time complexity is $\CO\left(2o_{d,1}+\sum_{v=2}^{V_3} o_{d,v-1}o_{d,v} \right)$, where $V_3$ is the number of layers of DNN, $o_{d,v}$ is the input size for layer $v$ \cite{TVTYuan}.
For the actor, as the stochastic policy is approximated by a CNN, the time complexity is identical with that of CNN in the meta critic.
Overall, the total time complexity of EMCL is calculated by $\CO\left(K(N+1)L_1+L_2\right)$, where $L_1 = \varrho_1+4mo_{e,1}$ and $L_2=\varrho_1 +o_{e,1}(4m+ \varsigma_1) + 2o_{d,1} + \sum_{v=2}^{V_1}o_{c,v-1} \varrho_v + \sum_{v=2}^{V_2} o_{e,v}(4o_{l,v-1} + \varsigma_v) +  \sum_{v=2}^{V_3} o_{d,v-1}o_{d,v}$.
 When the parameters of the learning model are determined, the complexity increases linearly with ${\bf P1}$'s input size, i.e., $K$ and $N$.

\section{Numerical Results} \label{section:3}

In the simulation, the parameter settings are similar as in \cite{DiBLEO2019, PopescuLinkbuget}. 
The adopted parameters for implementing EMCL are summarized in Table \ref{tab:paraMCDS}.
We compare the performance of the proposed EMCL algorithm with the following five benchmark algorithms: 
\begin{itemize}
\item OPT: optimal solution (B\&B).
\item ADMM-HEU: suboptimal solution (Alg. 1).
\item GRD: a greedy suboptimal algorithm proposed in \cite{GDuserscheduling}.
\item AC-DDPG:  a classic AC algorithm with deep deterministic policy gradient proposed in \cite{davidDDPG}.
\item AC-MAML:  AC with model-agnostic meta-learning proposed in \cite{KateMAMLAC}.
\end{itemize}
The first three provide benchmarks from an optimization perspective, while the last two compare with EMCL from a learning perspective.

\begin{table}[h]
\small
\vspace{-0mm}
\centering
\caption{Parameter setting}
\label{tab:paraMCDS}
\begin{tabular}{|c||c|}
\hline
Total number of GDs in network & 50-100 \\
\hline
Number of transmitters & 1 LEO, 1 BS and 2 TSTs \\
\hline
Time limitation $T$ & 10 time slots \\
\hline
Duration of time slot $\Phi$ & 0.1 s \\
\hline
Altitude of LEO  & 780 km \\
\hline
Transmit power of LEO & 100 W \\
\hline
Transmit power of BS & 40 W \\
\hline
Transmit power of TST & 2 W \\
\hline
Bandwidth for C-band & 20 MHz \\
\hline
Bandwidth for Ka-band & 400 MHz \\
\hline
Carrier frequency of C-Band & 4 GHz \\
\hline
Carrier frequency of Ka-Band & 30 GHz \\
\hline
Noise power spectral density & -174 dBm/Hz \\
\hline
Parameterized meta critic & HNN \\
\hline
Parameterized actor & CNN \\
\hline
Distribution of stochastic policy & Gaussian \\
\hline
Learning rate & 0.001 \\
\hline
Batch size & 128 \\
\hline
Memory size &  10,000 \\
\hline
Discount factor & 0.9 \\
\hline
Size of search space & \multirow{2}*{10}\\
in Wolpertinger policy &\\
\hline
Environment update interval & 200 time slots\\
\hline
\multirow{2}*{Software platform} &  Python 3.6 with\\
& TensorFlow 1.12.0\\
\hline
\end{tabular}
\end{table}

\subsection{Capability in Dealing with Dynamic Environments}

To verify the capability of the proposed EMCL in dealing with dynamic environments, Fig. \ref{fig:label05}-\ref{fig:label07} compare EMCL with AC-MAML and AC-DDPG in three dynamic scenarios.
In Fig. \ref{fig:label05}, we consider the first scenario with users' irregular access and departure, which can be disruptive to the typical statistical properties.
For instance, the adopted simulator generates user arrivals by following the Poisson distribution as the normal case, while it also periodically generates abnormal events (every 200 slots) with randomly large/small number of arrived users.
We update the environment information every 200 time slots.
From Fig. \ref{fig:label05}, both EMCL and AC-MAML are able to converge before each update, but EMCL saves 28.66\% recovery time and reduces 45.42\% objective value than AC-MAML, where we define a recovery time counting from the moment of dramatic performance degradation until the performance recoveries to the normal level.
For AC-DDPG, the convergence performance is inferior to the others, and fails to converge when updating occurs at the 200-th and 600-th time slot.
We remark that the case of user departure is easier to be adapted. 
Fewer users in the system reduce the problem dimension, and thus simplify the learning task, leading to a halved recovery time and flat curves between the 200-th and the 400-th slots in three algorithms. 
In contrast, it is more difficult to deal with the case of user arrival, referring to the large fluctuation after the 400-th slot, mainly due to lacking relevant new-user data and the exponentially increasing dimension.
We can observe that EMCL has strong capabilities in adapting to this difficult case and achieves more performance gains than the other two algorithms.

\begin{figure}[t]
\centering
\includegraphics[scale=0.46]{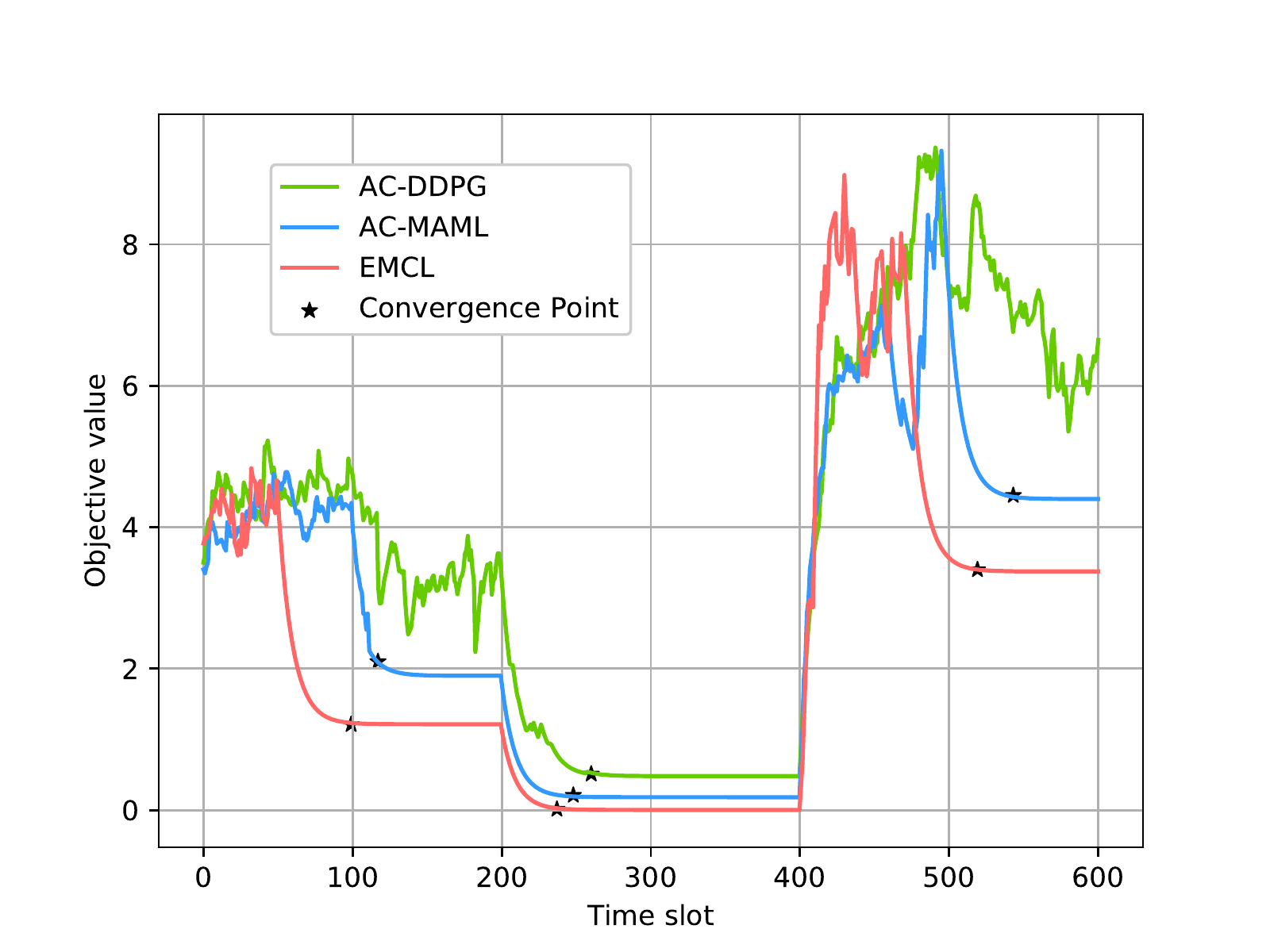}
\caption{Performance in adapting to dynamic scenario 1: \\user entry and leave.}
\label{fig:label05}
\end{figure} 

In Fig. \ref{fig:label06}, we evaluate the algorithms' capabilities in adapting to unforeseen dynamic demands. 
The simulator generates the volume of users' requested data by following the uniform distribution as the normal case.
Then, the distribution changes due to the abnormal bursty demands, e.g., switching from a low-speed voice call to a data-hungry HD video service, or vice versa.
In Fig. \ref{fig:label07}, we consider the channel states can undergo non-ideal large fluctuations, e.g., sharply deteriorated channel conditions due to the large obstacles or the rain/cloud blocks appearing in the transmission path.
Similarly to Fig. \ref{fig:label05}, we collect the updated environment information every 200 time slots.
From Fig. \ref{fig:label06} and Fig. \ref{fig:label07}, AC-DDPG has poor convergence performance, since AC-DDPG needs to re-train the learning model from scratch when the environment changes, leading to a slow adaptation, while EMCL and AC-MAML extract the meta-knowledge from multiple tasks to accelerate the convergence speed.
EMCL re-fits the learning model in a timely manner than AC-MAML.
This is because EMCL uses meta critic to guide the actor to adjust scheduling schemes more effectively in a dynamic environment, and the designed HNN and Wolpertinger mapping approach can improve the learning accuracy and efficiency in large discrete action spaces.

\begin{figure}[t]
\centering
\includegraphics[scale=0.46]{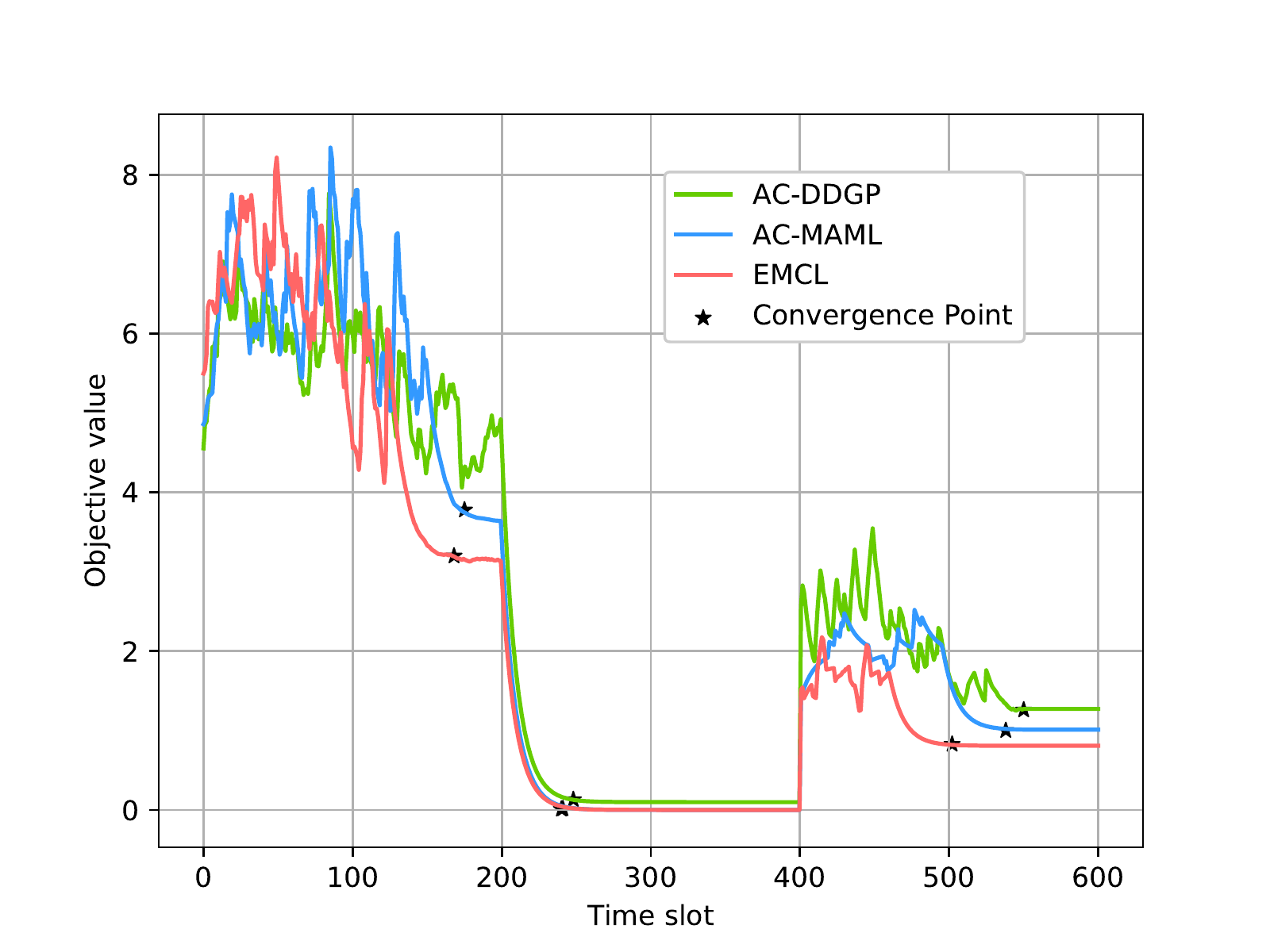}
\caption{Performance in adapting to dynamic scenario 2: \\bursty demands.}
\label{fig:label06}
\end{figure} 

\begin{figure}[t]
\centering
\includegraphics[scale=0.46]{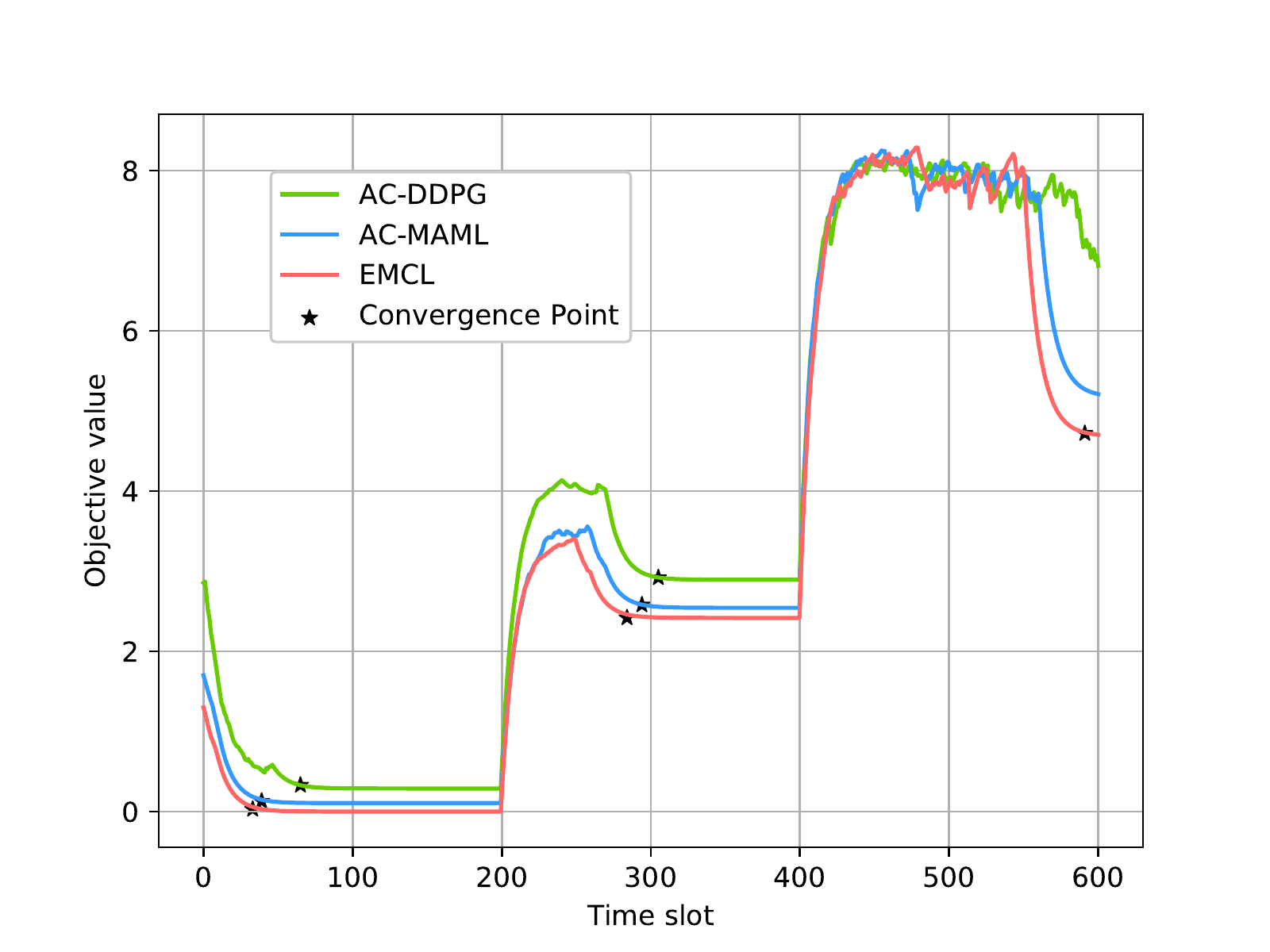}
\caption{Performance in adapting to dynamic scenario 3: \\unforeseen channel variations.}
\label{fig:label07}
\end{figure} 

Fig. \ref{fig:label08} further summarizes the average recovery time with respect to the numbers of GDs based on Fig. \ref{fig:label05}.
In general, the more GDs in the system, the longer the recovery time required to adapt to the new environment.
On average, EMCL saves 29.83\% and 13.49\% recovery time compared to AC-DDPG and AC-MAML, respectively, and the time-saving gain of EMCL becomes even larger when more GDs in the system.
In addition, we compare the EMCL algorithm with and without the Wolpertinger policy to demonstrate the effectiveness of the adopted action mapping method.
The recovery time of the latter is 10.11\% increased than the former but less than AC-DDPG and AC-MAML.
At the convergence, EMCL can decrease the average objective value by 30.36\% compared to EMCL without the Wolpertinger policy.

\begin{figure}[t]
\centering
\includegraphics[scale=0.46]{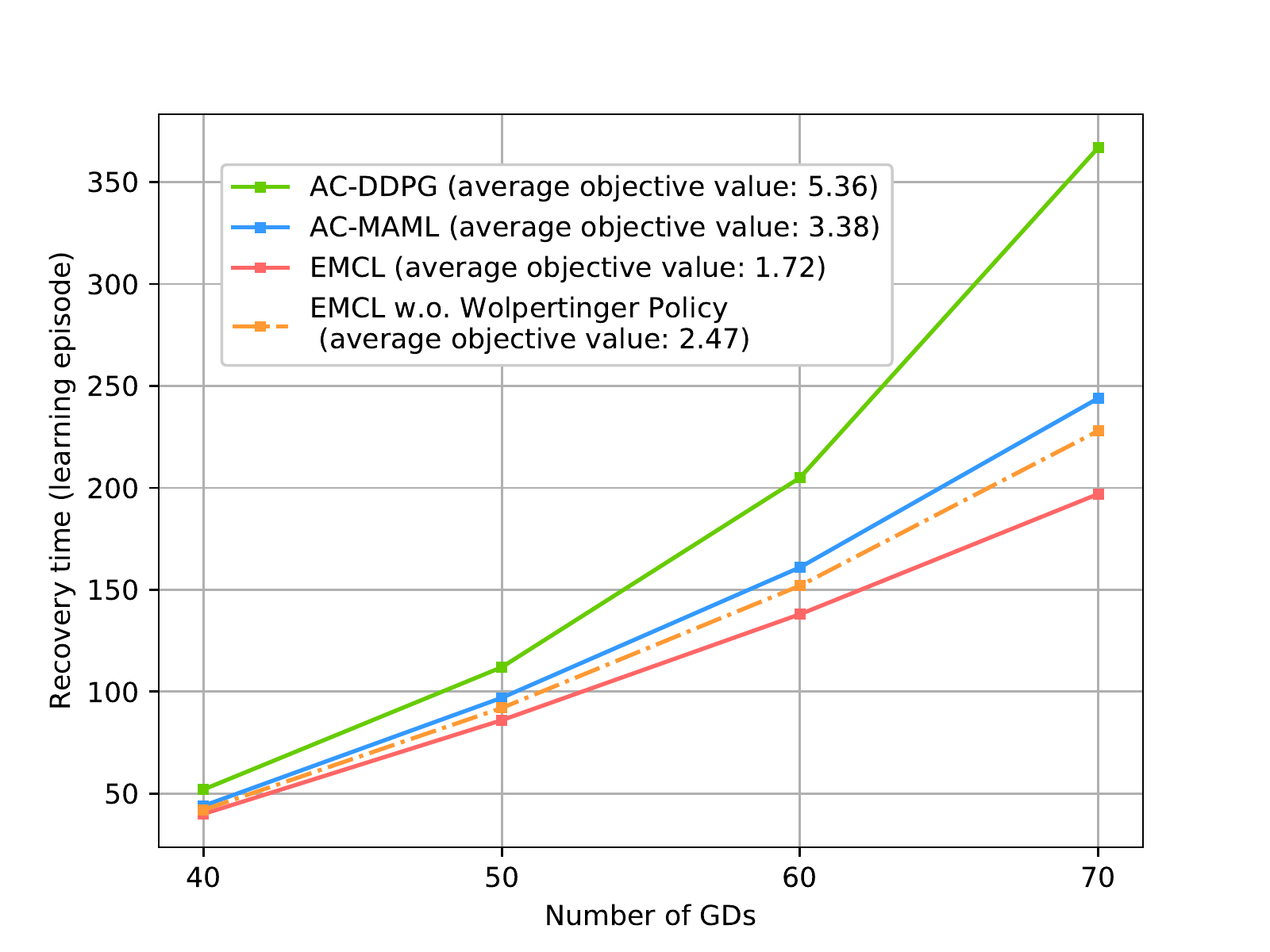}
\caption{Recovery time vs. number of users}
\label{fig:label08}
\end{figure} 

\subsection{Trade-Offs between Computational Time and Optimality}

To demonstrate EMCL's trade-off performance between approaching the  optimum (Fig. \ref{fig:label01}) and computational time (Fig. \ref{fig:label02}), we compare EMCL with five benchmarks.
In Fig. \ref{fig:label01}, we observe 50 environmental information updates and record the average objective values within each update cycle.
For AC-MAML and AC-DDPG, the average gaps to the optimum are 45.26\% and 57.23\%, respectively, while for EMCL, the average gap drops to 27.58\%.
The performance of EMCL is slightly better than ADMM-HEU, around 3.54\%.
For GRD, the average gap to the optimum is 74.15\%, which is inferior to the AC-based algorithms.

\begin{figure}[t]
\centering
\includegraphics[scale=0.46]{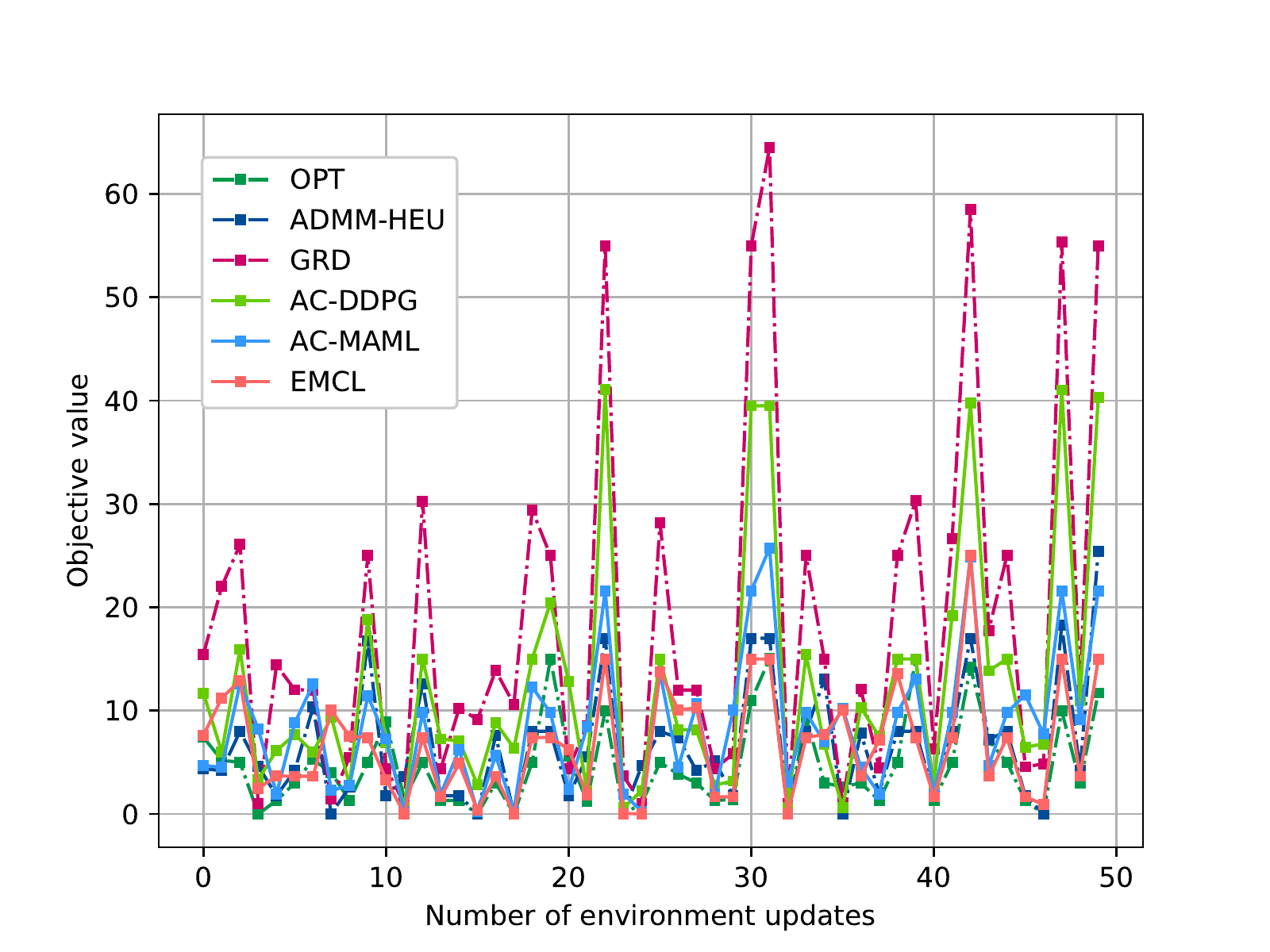}
\caption{Objective value vs. number of updates}
\label{fig:label01}
\end{figure} 

\begin{figure}[t]
\centering
\includegraphics[scale=0.46]{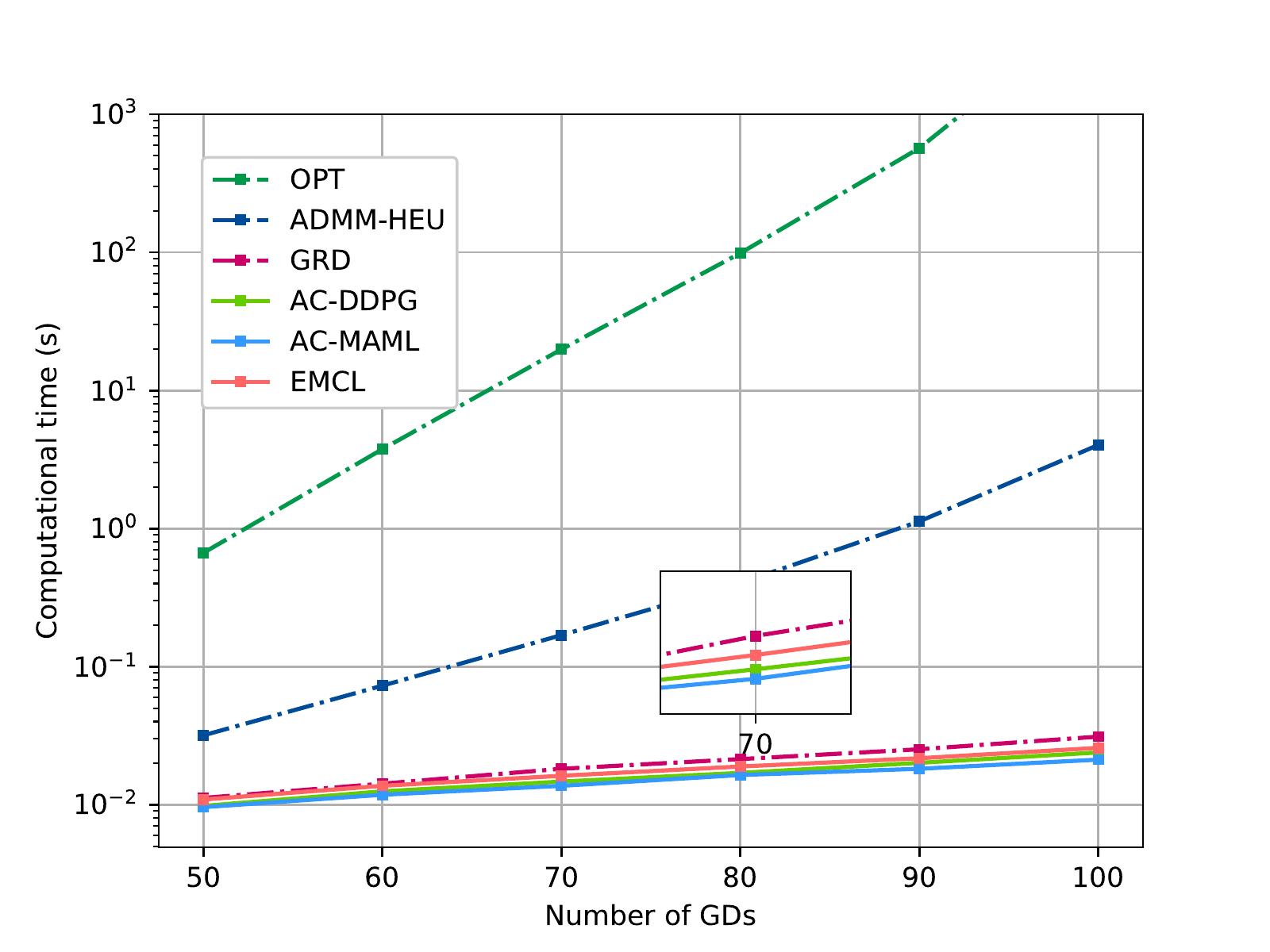}
\caption{Computational time vs. number of users}
\label{fig:label02}
\end{figure} 

Fig. \ref{fig:label02} compares the computational time with respect to the number of GDs.
OPT is the most time-consuming algorithm, as expected.
Compared to OPT, ADMM-HEU saves 98.14\% computational time by decomposing variables into multiple blocks and performing parallel computations.
The computational time in ECML, two AC algorithms, and GRD keep at the millisecond level, but the proposed EMCL achieves smaller gaps to the optimum, hence concludes the better trade-off performance of EMCL than other benchmarks.

\section{Conclusion}
We have investigated a resource scheduling problem in dynamic LEO-terrestrial communication systems to address the mismatch issue in a practical over-loaded scenario.
Due to the high computational time of the optimal algorithm and the proposed ADMM-HEU algorithm, we solve the problem from the perspective of DRL to obtain online solutions.
To enable the learning model to fast adapt to dynamic environments, we develop an EMCL algorithm that is able to handle the environmental changes in wireless networks, such as bursty demands, users' entry/leave, and abrupt channel change.
Numerical results show that, when encountering an environmental variation, EMCL consumes less recovery time to re-fit the learning model, compared to AC-DDPG and AC-MAML.
Furthermore, EMCL achieves a good trade-off between solutions quality and computation efficiency compared to offline and AC-based benchmarks.
An extension of the current work is to combine other techniques, e.g., continuous learning and behavior regularization, to further improve the sample efficiency and model adaptability.

\begin{appendices}
\section{Proof of Lemma 1}
\label{app:lemma1}
We relax all the binary variables of ${\bf P1}$ to  continuous variables $\hat{\textbf{x}}=[\hat{x}_{1,1},...,\hat{x}_{g,t},...,\hat{x}_{G,T}]^{\text{T}}$ and $\hat{\textbf{y}}=[\hat{y}_1,...,\hat{y}_k,...,\hat{y}_K]^{\text{T}}$, where $\hat{x}_{g,t},\hat{y}_k \in [0,1]$.
The relaxed objective function is written by:
\begin{align}
\label{eq:relaxedobj}
f(\hat{\textbf{x}}, \hat{\textbf{y}}) = \eta_0  \left(\mathbf{1}^{\text{T}}\hat{\textbf{y}} - K\right)^2 + \sum_{k \in \CK}\eta_{k}\left(\textbf{r}_k^{\text{T}}\hat{\textbf{x}} - D_k\right)^2,
\end{align}
where $\mathbf{1}=[1,...,1]^{\text{T}}$ and $\textbf{r}_k=[R_{k,1,1},...,R_{k,g,t},...,R_{k,G,T}]^{\text{T}}$.
We expand $ \left(\mathbf{1}^{\text{T}}\hat{\textbf{y}} - K\right)^2$ and $\left(\textbf{r}_k^{\text{T}}\hat{\textbf{x}} - D_k\right)^2$ as follows:
\begin{align}
\label{eq:item1obj}
&\left(\mathbf{1}^{\text{T}}\hat{\textbf{y}} - K\right)^2 = \hat{\textbf{y}}^{\text{T}}\textbf{E}\hat{\textbf{y}}-2D_k\textbf{1}_k^{\text{T}}\hat{\textbf{y}}+K,\\
\label{eq:item2obj}
&\left(\textbf{r}_k^{\text{T}}\hat{\textbf{x}} - D_k\right)^2 = \hat{\textbf{x}}^{\text{T}}\textbf{R}\hat{\textbf{x}}-2D_k\textbf{r}_k^{\text{T}}\hat{\textbf{x}}+D_k^2,
\end{align}
where $\textbf{E}$ is an all-ones matrix and
\begin{align}
\textbf{R} = 
\left[
\begin{array}{cccc}
R_{k,1,1}^2 & R_{k,1,1}R_{k,1,2} & \cdots & R_{k,1,1}R_{k,G,T} \\
R_{k,1,2}R_{k,1,1} & R_{k,1,2}^2 & \cdots & R_{k,1,2}R_{k,G,T} \\
\vdots & \vdots & \ddots & \vdots \\
R_{k,G,T}R_{k,1,1} & R_{k,G,T}R_{k,1,2} & \cdots & R_{k,G,T}^2 \\
\end{array}
\right].
\end{align}
Referring to the theorem of quadratic programming, a quadratic function is convex when its corresponding real symmetric matrix is positive semi-definite \cite{Nocedal}.
According to the definition, $\textbf{E}$ and $\textbf{R}$ are positive semi-definite matrices since, given an arbitrary vector $\textbf{v}=[v_1,...,v_{G\times T}]\neq \textbf{0}$, we can calculate $\textbf{v}^T\textbf{E}\textbf{v} = \left(\textbf{1}^{\text{T}}\textbf{v}\right)^2\geq 0$ and $\textbf{v}^T\textbf{R}\textbf{v} = \left(\textbf{r}^{\text{T}}\textbf{v}\right)^2\geq 0$ \cite{Nocedal}.
Therefore, $f(\hat{\textbf{x}}, \hat{\textbf{y}})$ is convex as it is the summation of $K+1$ convex functions.
Besides, the constraints Eq. (\ref{eq:p10})-(\ref{eq:p12}) are linear, hence the conclusion.

\section{Proof of Lemma 2}
\label{app:lemma2}
The objective of the learning agent is to find a policy $\pi(a|s_t)$ that maximizes the expected accumulated reward $\sum_{t=0}^{T}\gamma^{t}r_t$.
With $r_t$ in Eq. (\ref{eq:designreward}), we expand $\sum_{t=0}^{T}\gamma^{t}r_t$ as:
\begin{align}
\sum_{t=0}^{T}\gamma^{t}r_t =& \sum\limits_{k=0}^{K}\eta_k\left[\gamma \Delta_{k,0}^2 + \sum\limits_{t=1}^{T}(\gamma^{t+1}-\gamma^{t})\Delta_{k,t}^2-\gamma^{T} \Delta_{k,T}^2\right] \notag\\
\overset{\gamma=1}{=}& \sum\limits_{k=0}^{K}\eta_k\left( \Delta_{k,0}^2 - \Delta_{k,T}^2\right)\notag\\
=& \sum\limits_{k=0}^{K}\eta_k\Delta_{k,0}^2 - \eta_0 \left(\sum\limits_{k=1}^K \mathds{1}\left(b_{k,T}-D'_k\right)-K\right)^2 \notag\\
&- \sum\limits_{k=1}^{K} \eta_k\left(b_{k,T}-D_k \right)^2,
\end{align}
where $b_{k,T}=\sum_{t=1}^{T}R_{k,a_t,t}$. 
Thus, we can obtain the optimal policy $a^*_t\sim\pi^*(a|s_t)$ by solving the following problem:
\begin{align}
\label{eq:provlemma1}
\max\limits_{\pi(a|s_t)} -& \mathbb{E}_{\pi(a|s_t)}\left[\eta_0 \left(\sum\limits_{k=1}^K \mathds{1}\left(\sum_{t=1}^{T}R_{k,a_t,t}-D'_k\right)-K\right)^2\right. \notag\\
-& \left.\sum\limits_{k=1}^{K} \eta_k\left(\sum_{t=1}^{T}R_{k,a_t,t}-D_k \right)^2\right],
\end{align}
which is equivalent to the objective Eq. (\ref{eq:p1obj}), thus the conclusion.

\section{Proof of Lemma 3}
\label{app:lemma3}
Denote $Q(s, a_1),...,Q(s, a_M)$ as random variables $X_1,...,X_M$, where 
$X_{m'} = Q(s,a^*_{s})$ and $X_m\sim U(Q(s, a^*_{s})-\kappa, Q(s, a^*_{s})+\kappa),~\forall m\neq m'$.
Thus, $Q(s, a^*_{w})$ can be expressed as a random variable $\Psi=\max\{X_1,...,X_M\}$.
The cumulative distribution function of $\Psi$ is expressed as:
\begin{align}
\label{eq:provlemma201}
F_\Psi(\psi) =& \mathbb{P}[\Psi\leq \psi] = \mathbb{P}[\max\{X_1,...,X_M\}\leq \psi] \notag\\
=& \mathbb{P}[X_1\leq \psi]\mathbb{P}[X_2\leq \psi]...\mathbb{P}[X_M\leq \psi] \notag\\
=& F_{X_1}(\psi)F_{X_2}(\psi)...F_{X_M}(\psi)
\end{align}
For $m\neq m'$, based on the cumulative distribution function of uniform distribution, we can derive:
\begin{align}
\label{eq:provlemma202}
F_{X_m}(\psi) = &\frac{\psi-Q(s, a^*_{s})+\kappa}{2\kappa}, \notag\\
&\psi\in \left[Q(s, a^*_{s})-\kappa, Q(s, a^*_{s})+\kappa\right].
\end{align}
For $m=m'$, as $X_{m'} = Q(s,a^*_{s})$, the cumulative distribution function is:
\begin{align}
\label{eq:provlemma203}
F_{X_{m'}}(\psi) = \mathbb{P}[X_{m'}\leq \psi] = \left\lbrace
\begin{array}{cc}
1, & \psi\geq Q(s,a^*_{s}),\\
0, & \psi < Q(s,a^*_{s}).
\end{array}
\right.
\end{align}
By substituting Eq. (\ref{eq:provlemma202}) and Eq. (\ref{eq:provlemma203}) into Eq. (\ref{eq:provlemma201}),
\begin{align}
F_\Psi(\psi)\!=\!\left\lbrace
\begin{array}{cc}
\!\!\!\left(\frac{y-Q(s, a^*_{s})+\kappa}{2\kappa}\right)^{_{M\!-\!1}},&\!\!\!\!\psi \in [Q(s,a^*_{s}), Q(s, a^*_{s})+\kappa],\\
0, &\!\!\!\!\psi \in [Q(s,a^*_{s})-\kappa, Q(s, a^*_{s})).
\end{array}
\right.
\end{align} 
Then, the probability density function of $\Psi$ can be calculated by solving the first derivative:
\begin{align}
&f_\Psi(\psi) = [F_\Psi(\psi)]' \notag\\ 
= & \left\lbrace
\begin{array}{cc}
\frac{1}{2^{_{M-1}}}\delta(\psi-Q(s,a^*_{s})), & \psi = Q(s,a^*_{s}), \\
\frac{M-1}{2\kappa}\left(\frac{\psi-Q+\kappa}{2\kappa}\right)^{_{M-2}}, & Q(s,a^*_{s})< \psi \leq Q(s, a^*_{s})+\kappa, \\
0, & \text{otherwise},
\end{array}
\right.
\end{align}
where $\delta(\cdot)$ is Dirac function.
The expectation of $\Psi$ is:
\begin{align}
&\mathbb{E}\left[\Psi\right]=\mathbb{E}\left[Q(s,a^*_{w})\right]=\int_{Q(s,a^*_{s})}^{Q(s, a^*_{s})+\kappa}\psi f_\Psi(\psi)d\psi\notag\\
=& \frac{1}{2^{_{M-1}}}\int_{Q(s,a^*_{s})^{-}}^{Q(s, a^*_{s})^{+}}\psi\delta(y-Q(s,a^*_{s}))dy \notag\\
&+\int_{Q(s,a^*_{s})^{+}}^{Q(s, a^*_{s})+\kappa}\psi\frac{M-1}{2\kappa}\left(\frac{\psi-Q(s, a^*_{s})+\kappa}{2\kappa}\right)^{_{M-2}}d\psi \notag \\
=& \frac{Q(s,a^*_{s})}{2^{_{M-1}}} + Q(s,a^*_{s}) + \kappa - \frac{2\kappa}{M} - \frac{Q(s,a^*_{s})}{2^{_{M-1}}} + \frac{\kappa}{M\cdot2^{M-1}}\notag\\
=& Q(s,a^*_{s}) + \kappa\left(1-\frac{2(2^M-1)}{M\cdot 2^M}\right).
\end{align}
Thus the conclusion.
\end{appendices}

\end{document}